\documentclass[runningheads,a4paper,draft]{llncs}

\usepackage{amssymb,amsmath,mathtools}
\usepackage{alltt}
\usepackage{graphicx}
\usepackage[final,colorlinks=true,linktocpage=true,hyperfootnotes]{hyperref}
\usepackage{infer}
\usepackage{marvosym}
\usepackage{nicefrac}
\usepackage{stmaryrd}
\usepackage{tikz}
\usepackage[disable]{todonotes}
\usepackage{pgfplots}
\usepackage{url}
\usepackage{xspace}

\usetikzlibrary{arrows,decorations.pathmorphing,positioning,fit,trees,shapes,shadows,automata,calc}
\DeclareMathAlphabet{\mathpzc}{OT1}{pzc}{m}{it}

\urldef{\mailsa}\path|{benjamin.kaminski,katoen,matheja}@cs.rwth-aachen.de|
\newcommand{\keywords}[1]{\par\addvspace\baselineskip
\noindent\keywordname\enspace\ignorespaces#1}

\allowdisplaybreaks
\begin{document}

\hyphenation{va-ri-an-ces}

%%%%%     Programs     %%%%%

\newcommand{\SKIP}{\textnormal{\texttt{skip}}}
\newcommand{\EMPTY}{\textnormal{\texttt{empty}}}
\newcommand{\ABORT}{\textnormal{\texttt{diverge}}}
\newcommand{\HALT}{\textnormal{\texttt{halt}}}
\newcommand{\ASSIGN}[2]{\ensuremath{{#1} \coloneqq {#2}}}
\newcommand{\observesymbol}{\ensuremath{\textnormal{\texttt{observe}}}}
\newcommand{\OBSERVE}[1]{\ensuremath{\observesymbol~{#1}}}
\newcommand{\IFSYM}{\textnormal{\texttt{if}}}
\newcommand{\ELSESYM}{\textnormal{\texttt{else}}}
\newcommand{\ITE}[3]{\ensuremath{\IFSYM \, \left( {#1} \right) \, \left\{ {#2} \right\} \, \ELSESYM\, \left\{ {#3} \right\}}}
\newcommand{\PCHOICE}[3]{\ensuremath{\left\{ {#1} \right\} \mathrel{\left[ {#2} \right]} \left\{ {#3} \right\}}}
\newcommand{\COMPOSE}[2]{{#1}\textnormal{\texttt{;}}\allowbreak~ {#2}}
\newcommand{\WHILESYM}{\textnormal{\texttt{while}}}
\newcommand{\WHILE}[2]{\ensuremath{\WHILESYM \, \left( {#1} \right) \, \left\{ {#2} \right\}}}
\newcommand{\BOUNDEDWHILE}[3]{\ensuremath{\WHILESYM^{< #1} \, \left( {#2} \right) \, \left\{ {#3} \right\}}}

%%%%%     CONSTANTS     %%%%%

\newcommand{\imag}{\boldsymbol{i}}

%%%%%     FUNCTIONS / TRANSFORMERS     %%%%%

\newcommand{\To}{\rightarrow}

\newcommand{\wpsymbol}{{\sf wp}}
\renewcommand{\wp}[2]{\wpsymbol \left[ {#1} \right] \left( {#2}\right)}
\newcommand{\wpk}[3]{\wpsymbol^{#2} \left[ {#1} \right] \left( {#3}\right)}

\newcommand{\wlpsymbol}{{\sf wlp}}
\newcommand{\wlp}[2]{\wlpsymbol \left[ {#1} \right] \left( {#2}\right)}
\newcommand{\wlpk}[3]{\wlpsymbol^{#2} \left[ {#1} \right] \left( {#3}\right)}

\newcommand{\ertsymbol}{{\sf ert}}
\newcommand{\ert}[2]{\ertsymbol \left[ {#1} \right] \left( {#2}\right)}

\newcommand{\rtsymbol}{{\sf rt}}
\newcommand{\rt}[2]{\rtsymbol \left[ {#1} \right] \left( {#2}\right)}
\newcommand{\rtk}[3]{\rtsymbol^{#3} \left[ {#1} \right] \left( {#2}\right)}

\newcommand{\subst}[2]{\mathpunct{\left[ {#1} / {#2} \right]}}

\newcommand{\expsymbol}{\mathnormal{\mathsf{E}}}
\newcommand{\Exp}[3]{\expsymbol_{\llbracket #1 \rrbracket (#2) }\left( {#3} \right)}
\newcommand{\cExp}[2]{\expsymbol\left( {#1} \mathrel{|} {#2}\right)}

\newcommand{\covsymbol}{\mathnormal{\mathsf{Cov}}}
\newcommand{\Cov}[4]{\covsymbol_{\llbracket #1 \rrbracket (#2) }\left( {#3},\, {#4} \right)}
\newcommand{\cCov}[3]{\covsymbol\left( {#1},\, {#2} \mathrel{|} {#3} \right)}

\newcommand{\prsymbol}{\mathnormal{\mathsf{Pr}}}
\renewcommand{\Pr}[1]{\prsymbol\left( {#1} \right)}

\newcommand{\varsymbol}{\mathnormal{\mathsf{Var}}}
\newcommand{\Var}[3]{\varsymbol_{\llbracket #1 \rrbracket (#2) } \left( {#3} \right)}
\newcommand{\rtvarsymbol}{\mathnormal{\mathsf{RTVar}}}
\newcommand{\RtVar}[2]{\rtvarsymbol_{\llbracket #1 \rrbracket (#2)}}
\newcommand{\cVar}[2]{\varsymbol\left( {#1} \mathbin{|} {#2} \right)}

%%%%%     SETS     %%%%%
\newcommand{\PProgs}{\ensuremath{{\mathbb{P}}}}
\newcommand{\CharFuns}{\ensuremath{\Phi}}
\newcommand{\Rats}{\ensuremath{\mathbb{Q}}}
\newcommand{\PosRats}{\ensuremath{\mathbb{Q}_{\geq 0}}}
\newcommand{\Reals}{\ensuremath{\mathbb{R}}}
\newcommand{\PosReals}{\ensuremath{\mathbb{R}_{\geq 0}}}
\newcommand{\PosRealsInf}{\ensuremath{\mathbb{R}_{\geq 0}^{\infty}}}
\newcommand{\Complexs}{\ensuremath{\mathbb{C}}}
\newcommand{\ComplexsInf}{\ensuremath{\mathbb{C}^{\infty}}}
\newcommand{\Nats}{\ensuremath{\mathbb{N}}}
\newcommand{\E}{\ensuremath{\mathbb{E}}}
\newcommand{\BE}{\ensuremath{\mathbb{E}_{\leq 1}}}
\newcommand{\Etau}{\ensuremath{\mathbb{E}_{\tau}}}
\newcommand{\T}{\ensuremath{\mathbb{T}}}
\newcommand{\VPairs}{\ensuremath{\mathbb{V}}}
\newcommand{\Vars}{\ensuremath{\mathbb{V}}}
\newcommand{\States}{\ensuremath{\mathbb{S}}}
\newcommand{\Statestau}{\ensuremath{\mathbb{S}_\tau}}

%%%%%     PROBLEMS AND REDUCTIONS     %%%%%

\newcommand{\Problem}[1]{\mathcal{#1}}
\newcommand{\cProblem}[1]{\overline{\text{\footnotesize{$\Problem{#1}$}}}}

\newcommand{\HP}{\Problem{H}\xspace}
\newcommand{\cHP}{\cProblem{H}\xspace}
\newcommand{\UHP}{\Problem{UH}\xspace}
\newcommand{\cUHP}{\cProblem{UH}\xspace}
\newcommand{\COF}{\Problem{COF}\xspace}
\newcommand{\cCOF}{\cProblem{COF}\xspace}
\newcommand{\AST}{\ensuremath{\Problem{AST}}\xspace}
\newcommand{\cAST}{\ensuremath{\cProblem{AST}}\xspace}
\newcommand{\UAST}{\ensuremath{\Problem{U\hspace{-1pt}AST}}\xspace}
\newcommand{\cUAST}{\ensuremath{\cProblem{U\hspace{-1pt}AST}}\xspace}
\newcommand{\PAST}{\ensuremath{\Problem{P\hspace{-2pt}AST}}\xspace}
\newcommand{\cPAST}{\ensuremath{\cProblem{P\hspace{-2pt}AST}}\xspace}
\newcommand{\UPAST}{\ensuremath{\Problem{UP\hspace{-2pt}AST}}\xspace}
\newcommand{\cUPAST}{\ensuremath{\cProblem{UP\hspace{-2pt}AST}}\xspace}
\newcommand{\EXP}{\ensuremath{\Problem{E\hspace{-1pt}XP}}}
\newcommand{\LEXP}{\ensuremath{\Problem{LE\hspace{-1pt}XP}}}
\newcommand{\REXP}{\ensuremath{\Problem{RE\hspace{-1pt}XP}}}

\newcommand{\LCOVAR}{\ensuremath{\Problem{LCOV\hspace{-.5ex}AR}}}
\newcommand{\RCOVAR}{\ensuremath{\Problem{RCOV\hspace{-.5ex}AR}}}
\newcommand{\LVAR}{\ensuremath{\Problem{LV\hspace{-.5ex}AR}}}
\newcommand{\RVAR}{\ensuremath{\Problem{RV\hspace{-.5ex}AR}}}
\newcommand{\COVAR}{\ensuremath{\Problem{COV\hspace{-.5ex}AR}}}
\newcommand{\VAR}{\ensuremath{\Problem{V\hspace{-.5ex}AR}}}
\newcommand{\INFCOVAR}{\ensuremath{{}^{\infty}\Problem{COV\hspace{-.5ex}AR}}}
\newcommand{\INFVAR}{\ensuremath{{}^\infty\Problem{V\hspace{-.5ex}AR}}}
\newcommand{\LRTVAR}{\ensuremath{\Problem{L\hspace{-.25ex}R\hspace{-.25ex}TV\hspace{-.5ex}AR}}}
\newcommand{\RRTVAR}{\ensuremath{\Problem{R\hspace{-.25ex}R\hspace{-.25ex}TV\hspace{-.5ex}AR}}}
\newcommand{\RTVAR}{\ensuremath{\Problem{R\hspace{-.25ex}TV\hspace{-.5ex}AR}}}
\newcommand{\INFRTVAR}{\ensuremath{{}^\infty\Problem{R\hspace{-.25ex}TV\hspace{-.5ex}AR}}}

\newcommand{\leqm}{\mathrel{\:{\leq}_{\mathrm{m}}}}
\newcommand{\leqT}{\mathrel{\:{\leq}_{\mathrm{T}}}}
\newcommand{\equivm}{\mathrel{\:{\equiv}_{\mathrm{m}}}}
\newcommand{\equivT}{\mathrel{\:{\equiv}_{\mathrm{T}}}}

\newcommand{\LLL}[1]{\ensuremath{\mathrm{L}\big(#1\big)}\xspace}
\newcommand{\RRR}[1]{\ensuremath{\mathrm{R}\big(#1\big)}\xspace}

\newcommand{\lfp}{\ensuremath{\mathop{\mathsf{lfp}}}}
\newcommand{\gfp}{\ensuremath{\mathop{\mathsf{gfp}}}}

%%%%%    OPERATIONAL SEMANTICS %%%%%%
\newcommand{\MCSYMBOL}{\ensuremath{\mathpzc{M}}} % symbol for Markov Chains
\newcommand{\OPMC}[3]{\ensuremath{\MCSYMBOL_{#2}^{#3}\left[ #1 \right]}} % operational MC of program #1, initial state #2, continuation #3
\newcommand{\MCSTATES}{\ensuremath{\mathcal{S}}}
\newcommand{\MCTRANS}{\ensuremath{\mathbf{P}}}
\newcommand{\MCREW}{\ensuremath{rew}}
\newcommand{\TERM}{\ensuremath{\downarrow}}
\newcommand{\SINK}{\ensuremath{\mathrm{sink}}}
\newcommand{\OBSERVEFAIL}{\ensuremath{\mbox{\Lightning}}}
\newcommand{\MCSTATE}[1]{\ensuremath{\langle\, #1 \,\rangle}}

\newcommand{\MCINFRULE}[3]{\ensuremath{\infrule{#1}{#2}~\mathrm{\small \textsf{[#3]}}}}
\newcommand{\MCTRANSITION}[3]{\ensuremath{\MCSTATE{#1} \xrightarrow{#3} \MCSTATE{#2}}}
\newcommand{\MCPROB}[2]{\ensuremath{\mathrm{Pr}^{#1} \{ #2 \}}}
\newcommand{\MCREACH}[1]{\ensuremath{\Diamond #1}}

\newcommand{\ExpRew}[2]{\ensuremath{\mathrm{\textsf{ExpRew}}^{#1}\left(#2\right)}}
\newcommand{\CondExpRew}[3]{\ensuremath{\mathnormal{\mathsf{CExpRew}}^{#1}\left(#2 ~|~ #3\right)}}

%%%%%     DOT CUP     %%%%%
\makeatletter
\def\moverlay{\mathpalette\mov@rlay}
\def\mov@rlay#1#2{\leavevmode\vtop{%
   \baselineskip\z@skip \lineskiplimit-\maxdimen
   \ialign{\hfil$\m@th#1##$\hfil\cr#2\crcr}}}
\newcommand{\charfusion}[3][\mathord]{
    #1{\ifx#1\mathop\vphantom{#2}\fi
        \mathpalette\mov@rlay{#2\cr#3}
      }
    \ifx#1\mathop\expandafter\displaylimits\fi}
\makeatother

\newcommand{\dotcup}{\charfusion[\mathrel]{\cup}{\raisebox{1pt}{$\boldsymbol{\cdot}$}}}

\mainmatter  % start of an individual contribution

\title{Inferring Covariances for\\ Probabilistic Programs\thanks{This work was supported by the Excellence Initiative of the German federal and state government.}}

\titlerunning{Inferring Covariances for Probabilistic Programs}

\author{Benjamin Lucien Kaminski \and Joost-Pieter Katoen \and Christoph Matheja}
\authorrunning{Kaminski \and Katoen \and Matheja}

\institute{Software Modelling and Verification Group\\
RWTH Aachen University
\mailsa
}


\toctitle{Inferring Covariances for Probabilistic Programs}
\tocauthor{Kaminski, Katoen, and Matheja}
\maketitle

\begin{abstract}
We study weakest precondition reasoning about  the (co)va{\-}ri{\-}an{\-}ce of outcomes and the variance of run--times of probabilistic programs with conditioning.
For outcomes, we show that approximating (co)variances is computationally more difficult than approximating expected values.
In particular, we prove that computing both lower and upper bounds for (co)variances is $\Sigma_2^0$--complete.
As a consequence, neither lower nor upper bounds are computably enumerable.
We therefore present invariant--based techniques that \emph{do} enable enumeration of both upper and lower bounds, once appropriate invariants are found.
Finally, we extend this approach to reasoning about run--time variances.
%weakest precondition $\:\cdot\:$ computational hardness}% $\:\cdot\:$ skewness $\:\cdot\:$ kurtosis}
\end{abstract}

\section{Introduction}
Probabilistic programs describe manipulations on uncertain data in a succinct way.
They are normal--looking programs describing how to obtain a distribution
over the outputs.
Using mostly standard programming language constructs, a probabilistic
program transforms a prior distribution into a posterior distribution.
Probabilistic programs provide a structured means to describe e.g.,
Bayesian networks (from AI), random encryption (from security), or
predator--prey models (from biology)~\cite{DBLP:conf/icse/GordonHNR14}
succinctly.

The posterior distribution of a program is mostly determined by
approximate means such as Markov Chain Monte Carlo (MCMC) sampling using
(variants of) the well--known Metropolis--Hasting approach.
This yields estimates for various measures of interest, such as expected
values, second moments, variances, covariances, and the like.
Such estimates typically come with weak guarantees in the form of
confidence intervals, asserting that with a certain confidence the
measure has a certain value.
In contrast to these weak guarantees, we aim at the \emph{exact}
inference of such measures and their bounds.
We hereby focus both on correctness and on run--time analysis of
probabilistic programs.
Put shortly, we are interested in obtaining \emph{quantitative}
statements about the possible outcomes of programs well as their run times.

This paper studies reasoning about the (co)variance of outcomes and the
variance of run--times of probabilistic programs.
Our programs support sampling from discrete probability distributions,
conditioning on the outcomes of experiments by
observations~\cite{DBLP:conf/icse/GordonHNR14}, and unbounded
while--loops\footnote{This contrasts MCMC--based analysis, as this
is restricted to bounded programs.}.
In the first part of the paper, we study the \emph{theoretical
complexity} of obtaining (co)variances on outcomes.
We show that obtaining bounds on (co)variances is computationally more
difficult than for expected values.
In particular, we prove that computing both upper \emph{and lower} bounds for
(co)variances of program outcomes is $\Sigma_2^0$--complete, thus \emph{not recursively enumerable}.
This contrasts the case for expected values where lower bounds \emph{are
recursively enumerable}, while only upper bounds are
$\Sigma^0_2$--complete~\cite{hardness}.
We also show that determining the precise values of (co)variances as
well as checking whether the (co)variance is infinite are both
$\Pi^0_2$--complete.
These results rule out analysis techniques based on finite
loop--unrollings as complete approaches for reasoning about the
covariances of outcomes of probabilistic programs.

In the second part of the paper, we therefore develop a weakest
precondition reasoning technique for obtaining covariances on outcomes
and variances on run--times.
As with deductive reasoning for ordinary sequential programs, the crux
is to find suitable loop--invariants.
We present a couple of invariant--based proof rules that provide a sound
and complete method to computably enumerate both upper and lower bounds on
covariances, once appropriate invariants are found.
We establish similar results for variances of the run--time of programs.
The results of this paper extend McIver and Morgan’s approach for
obtaining expectations of probabilistic programs~\cite{mcivers},
recent techniques for expected run--time analysis~\cite{esop16}, and complement results on termination
analysis~\cite{hardness,luis}.

Some proofs had to be omitted due to lack of space. They can be found in an extended version of this paper~\cite{technicalReport}.

\section{Preliminaries}
We study approximating the covariance of two random variables (ranging over program states) after successful termination of a probabilistic program on a given input state.
Our development builds upon the \emph{conditional probabilistic guarded command language (cpGCL)}~\cite{mfps}---an extension of Dijkstra's guarded command language \cite{dijkstra} endowed with probabilistic choice and conditioning constructs.
\begin{definition}[cpGCL~\cite{mfps}]
Let $\Vars$ be a finite \emph{set of program variables}\footnote{We restrict ourselves to a finite set of program variables for reasons of cleanness of the presentation. In principle, a countable set of program variables could be allowed.}.
Then the \emph{set of programs in cpGCL}, denoted $\PProgs$,  adheres to the grammar
\begin{align*}
\PProgs ~~\Coloneqq~~ 	& \SKIP \;\big|\; \EMPTY \;\big|\; \ABORT \;\big|\; \HALT \;\big|\; \ASSIGN{x}{E} \;\big|\; \COMPOSE{\PProgs}{\PProgs}	\;\big|\; \ITE{B}{\PProgs}{\PProgs} \\
 						& \;\big|\; \PCHOICE{\PProgs}{p}{\PProgs} \;\big|\; \WHILE{B}{\PProgs} \;\big|\; \OBSERVE{B} ~,
\end{align*}
where $x \in \Vars$, $E$ is an arithmetical expression over \Vars, $p \in [0,\, 1] \cap \mathbb{Q}$ is a rational probability, and $B$ is a Boolean expression over arithmetic expressions over $\Vars$. 

If a program $C$ contains neither a probabilistic choice $\PCHOICE{C'}{p}{C''}$ nor an $\observesymbol$--statement, we say that $C$ is \emph{non--probabilistic}.
\end{definition}
We briefly go over the meaning of the language constructs. Furthermore, we assign each statement an execution time in order to reason about the \emph{run--time} of programs.
$\SKIP$ ($\EMPTY$) does nothing---i.e.\ does not alter the current variable valuations---and consumes one (no) unit of time.
$\ABORT$ is syntactic sugar for the certainly non--terminating program $\WHILE{\texttt{true}}{\SKIP}$.
$\HALT$ consumes no unit of time and halts program execution immediately (even when encountered inside a loop).
It represents an \emph{improper} termination of the program.
$\ASSIGN x E$, $\COMPOSE{C_1}{C_2}$, $\ITE{B}{C_1}{C_2}$, and $\WHILE{B}{C'}$ are standard variable assignment, sequential composition, conditional choice, and while--loop constructs.
Assignments and guard evaluations % for the conditional choice and the while--loop 
consume one unit of time.

$\PCHOICE{C_1}{p}{C_2}$ is a probabilistic choice construct:
With probability $p$ the program $C_1$ is executed and with probability $1-p$ the program $C_2$ is executed.
Flipping the $p$--coin itself consumes one unit of time.
$\OBSERVE B$ is the conditioning construct.
Whenever in the execution of a program, an $\OBSERVE{B}$ is encountered, such that the current variable valuation satisfies the guard $B$, nothing happens except that one unit of time is being consumed.
If, however, an $\OBSERVE{B}$ is encountered along an execution trace that occurs with probability $q$, such that $B$ is \emph{not} satisfied, this trace is blocked as it is considered an \emph{undesired execution}.
The probabilities of the remaining execution traces are then conditioned to the fact that this undesired trace was not encountered, i.e.\ the probabilities of the remaining execution traces are renormalized by $1 - q$.
We refer to encountering such an undesired execution as an \emph{observation violation}.
For more details on conditioning and its semantics, see~\cite{mfps}.

Notice that we do not include non--deterministic choice constructs (as opposed to probabilistic choice construct) in our language, as we would then run into similar problems as in \cite[Section 6]{mfps} in the presence of conditioning.
\begin{example}[Conditioning inside a Loop]
\label{firstexample}
	Consider the following loop:
\begin{alltt}
while (\(c=1\))\{ \{\ASSIGN{c}{0}\} [0.5] \{\ASSIGN{x}{x+1}\}; \OBSERVE{c=1 \vee x\,\,\textnormal{is\,\,odd}} \}
\end{alltt}
	Without the $\observesymbol$--statement, this loop would generate a geometric distribution on $x$.
	By considering the $\observesymbol$--statement, this distribution is conditioned to the fact that after termination $x$ is odd. \hfill$\triangle$
\end{example}
Given a probabilistic program $C$, an initial state $\sigma$, and a random variable $f$ mapping program states to positive reals, we could now ask: 
What is the \emph{conditional} expected value of $f$ after proper termination of program $C$ on input $\sigma$, \emph{given that no observation was violated during the execution}?
An answer to this question is given by the conditional weakest pre--expectation calculus introduced in \cite{mfps}.
For summarizing this calculus, we first formally characterize the random variables $f$, commonly called \emph{expectations} \cite{mcivers}:
\begin{definition}[Expectations \cite{mcivers,mfps}]
Let $\States = \{ \sigma ~|~ \sigma \colon \Vars \To \Rats \}$, where $\Rats$ is the set of rational numbers, be the \emph{set of program states}.\footnote{Notice that $\States$ is countable and computably enumerable as $\Vars$ is finite.
}  
Then the \emph{set of expectations} %, denoted $\pmb{\E}$,
is defined as $\E = \left\{f ~\middle|~ f\colon \States \To \PosRealsInf\right\}$, and the \emph{set of bounded expectations}
%, denoted $\pmb{\BE}$, 
is defined as $\BE = \left\{f ~\middle|~ f\colon \States \To [0,\, 1]\right\}$.
A \emph{complete partial order $\preceq$} on both ${\E}$ and ${\BE}$ is given by $f_1 \preceq f_2$ iff $\forall\, \sigma \in \States \colon f_1(\sigma) \leq f_2(\sigma)$.
\end{definition}
The \emph{weakest (liberal) pre--expectation transformer $\wpsymbol\colon \PProgs \to (\E \to \E)$} ($\wlpsymbol\colon \PProgs \to (\BE \to \BE)$) is defined according to \autoref{table:transrules} (middle column).
\begin{table}[t!]
\begin{center}
	\scalebox{.956}{
	\begin{tabular}{l@{\hspace{2em}}l@{\hspace{2em}}l}
		\hline\\[-7pt]
		$\boldsymbol{C}$ 			& $\boldsymbol{\wp{C}{f}}$						& $\boldsymbol{\rt{C}{t}}$								\\[2pt]
		\hline\hline\\[-7pt]
		$\SKIP$ 		  			& $f$											& $t\subst{\tau}{\tau + 1}$								\\[4pt]
		$\EMPTY$				& $f$											& $t$												\\[4pt]
		$\ABORT$				& $\boldsymbol{0}$								& $\boldsymbol{\infty}$								\\[4pt]
		$\HALT$					& $\boldsymbol{0}$								& $\boldsymbol{0}$									\\[4pt]
		$\ASSIGN{x}{E}$			& $f\subst{x}{E}	$								& $t[x/E,\, \tau/\tau + 1]$								\\[4pt]
		$\COMPOSE{C_1}{C_2}$		& $\wpsymbol[C_1] \circ \wp{C_2}{f}$				& $\rtsymbol[C_1] \circ \rt{C_2}{t}$						\\[4pt]
		$\ITE{B}{C_1}{C_2}$			& $[B] \cdot \wp{C_1}{f}$							& $\big([B] \cdot \rt{C_1}{t} $							\\[4pt]
								& $\quad + [\neg B] \cdot \wp{C_2}{f}$				& $\quad + [\neg B] \cdot \rt{C_2}{t}\big)\subst{\tau}{\tau + 1}$	\\[4pt]
		$\PCHOICE{C_1}{p}{C_2}$	& $p \cdot \wp{C_1}{f}$							& $\big(p \cdot \rt{C_1}{t}$								\\[4pt]
								& $\quad + (1 - p) \cdot \wp{C_2}{f}$				& $\quad + (1 - p) \cdot \rt{C_2}{t}\big)\subst{\tau}{\tau + 1}$	\\[4pt]
		$\WHILE{B}{C'}$			& $\lfp X.~ [\neg B] \cdot f $						& $\lfp X.~ \big([\neg B] \cdot t$							\\[4pt]
								& $\quad + [B] \cdot \wp{C'}{X}$					& $\quad + [B] \cdot \rt{C'}{X}\big)\subst{\tau}{\tau + 1}$		\\[4pt]
		$\OBSERVE{B}$ 			& $[B] \cdot f$									& $[B] \cdot t\subst{\tau}{\tau+1}$						\\[6pt]
		\hline\\[-7pt]
		$\boldsymbol{C}$ 			& $\boldsymbol{\wlp{C}{f}}$						\\[2pt]
		\hline\hline\\[-7pt]
		$\ABORT$				& $\boldsymbol{1}$								\\[4pt]
		$\HALT$					& $\boldsymbol{1}$								\\[4pt]
		%
%		$\COMPOSE{C_1}{C_2}$		& $\wlpsymbol[C_1] \circ \wlp{C_2}{f}$				\\[4pt]
%		$\ITE{B}{C_1}{C_2}$			& $[B] \cdot \wlp{C_1}{f} + [\neg B] \cdot \wlp{C_2}{f}$	\\[4pt]
%		$\PCHOICE{C_1}{p}{C_2}$	& $p \cdot \wlp{C_1}{f} + (1 - p) \cdot \wlp{C_2}{f}$		\\[4pt]
		%
		$\WHILE{B}{C'}$			& \multicolumn{2}{l}{$\gfp X.~ [\neg B] \cdot f + [B] \cdot \wlp{C'}{X}$}		\\[4pt]
		\hline
	\end{tabular}
	}
\end{center}
\caption{Definition of $\wpsymbol$, $\wlpsymbol$, and $\rtsymbol$. 
$\subst{x}{E}$ is a syntactic replacement with $f\subst{x}{E}(\sigma)\allowbreak =\allowbreak f(\sigma[x \mapsto \sigma(E)])$.
$[B]$ is the indicator function of $B$ with $[B](\sigma) = 1$ if $\sigma \models B$, and $[B](\sigma) = 0$ otherwise.
$F \circ H (f)$ is the functional composition of $F$ and $H$ applied to $f$. %, i.e.\ first applying $H$ to $f$ and then applying $F$ to the result. 
$\lfp X.~F(X)$ ($\gfp X.~F(X)$) is the least (greatest) fixed point of $F$ with respect to $\preceq$.
Definitions of $\wlpsymbol$ for the other language constructs
%$\SKIP$, $\EMPTY$, $\ASSIGN x E$, and $\OBSERVE B$ 
are as for $\wpsymbol$ and thus omitted.}
\label{table:transrules}
\vspace{-2\baselineskip}
\end{table}
%
%
%Indeed, these transformers are---from a domain theoretic point of view---well--behaved, as they enjoy the following two properties:
%%
%%
%\begin{theorem}[Continuity and Monotonicity of $\boldsymbol{\wpsymbol}$ and $\boldsymbol{\wlpsymbol}$  \cite{mfps}]
%For all $C \in \PProgs$, all $S \subseteq \E$, and all $f_1,f_2 \in \E$ holds:
%\begin{enumerate}
%\item[\emph{1.}] \emph{Continuity:} $\wp{C}{\sup_{f \in S} f} = \sup_{f \in S} \wp{C}{f}$ and $\wlp{C}{\sup_{f \in S} f}\allowbreak = \sup_{f \in S} \wlp{C}{f}$.
%\item[\emph{2.}] \emph{Monotonicity:} $f_1 \sqsubseteq f_2$ implies $\wp{C}{f_1} \sqsubseteq \wp{C}{f_2}$ and $\wlp{C}{f_1} \sqsubseteq \wlp{C}{f_2}$.
%\end{enumerate}
%\end{theorem}
%%
%%
By means of these two transformers, we can give an answer to the question posed above:
Namely, the fraction 
$\nicefrac{\wp{C}{f}(\sigma)}{\wlp{C}{\boldsymbol{1}}(\sigma)}$ is indeed the 
conditional expected value of $f$ after termination of $C$ on input $\sigma$, 
given that no observation was violated during $C$'s execution~\cite{mfps}.
Consequently, we define:
\begin{definition}[Conditional Expected Values \cite{mfps}]
\label{def:cond-ev}
Let $C \in \PProgs$, $\sigma \in \States$, and $f \in \E$.
Then the \emph{conditional expected value} of $f$ after executing $C$ on input $\sigma$ given that no observation was violated %, denoted $\boldsymbol{\Exp{C}{\sigma}{f}}$, 
is defined as\footnote{We make use of the convention that $\frac 0 0 = 0$. 
Note that since our probabilistic choice is a discrete choice and our language does not support sampling from continuous distributions, the problematic case of ``$\frac 0 0$" can only occur if executing $C$ on input $\sigma$ will result in a violation of an observation with probability 1.
}
\belowdisplayskip=0pt
\begin{align*}
	\Exp{C}{\sigma}{f} ~=~ \frac{\wp{C}{f}(\sigma)}{\wlp{C}{\boldsymbol{1}}(\sigma)}~.
\end{align*}
\normalsize
\end{definition}
Having the definition for conditional expected values readily available, we can now turn towards defining the conditional (co)variance of a (two) random variables.
We simply translate the textbook definition to our setting:
\begin{definition}[Conditional (Co)variances] \label{def:covariances}
\label{def:cond-cov}
Let $C \in \PProgs$, $\sigma \in \States$, and $f,g \in \E$.
Then the \emph{conditional covariance} of the two random variables $f$ and $g$ after executing $C$ on input $\sigma$, given that no observation was violated 
%, denoted $\boldsymbol{\Cov{C}{\sigma}{f}{g}}$,
is defined as
\begin{align*}
	\Cov{C}{\sigma}{f}{g} ~=~ \Exp{C}{\sigma}{f \cdot g} - \Exp{C}{\sigma}{f} \cdot \Exp{C}{\sigma}{g}~.
\end{align*}
The \emph{conditional variance} of the single random variable $f$ after executing $C$ on input $\sigma$, given that no observation was violated %, denoted $\boldsymbol{\Var{C}{\sigma}{f}}$, 
is defined as the conditional covariance of $f$ with itself, i.e.\ $\Var{C}{\sigma}{f} = \Cov{C}{\sigma}{f}{f}$.
%If we do not condition to any particular event\footnote{This is the same as conditioning to some certain event which has probability 1.}, we simply speak of the \textbf{covariance}, denoted $\boldsymbol{\Cov{C}{\sigma}{X}{Y}}$, and the \textbf{variance}, denoted $\boldsymbol{\Var{C}{\sigma}{X}}$.
\end{definition}

\section{Computational Hardness of Computing (Co)variances}
In this section, we will investigate the computational hardness of computing upper and lower bounds for conditional (co)variances.
The results will be stated in terms of levels in the arithmetical hie{\-}rar{\-}chy---a concept we first briefly recall:
\begin{definition}[The Arithmetical Hierarchy \textnormal{\textbf{\cite{kleeneNF,odifreddi1}}}]
\label{remarithmetic}
For every $n \in \Nats$, the \emph{class $\Sigma_n^0$} is defined as $\Sigma_n^0 = \big\{ \Problem A ~\big|~ \Problem A = \big\{  x ~\big|~ \exists y_1\, \forall y_2\, \exists y_3\, \cdots\, \exists/\forall y_n\colon~ ( x,\, y_1,\, y_2,\allowbreak\, y_3,\, \ldots,\allowbreak\, y_n) \in \Problem R\big\},\, \Problem R$ \textnormal{is a decidable relation}$\big\}$ and the \emph{class $\Pi_n^0$} is defined as $\Pi_n^0 = \big\{ \Problem A ~\big|~ \Problem A = \big\{ x ~\big|~ \forall y_1\, \exists y_2\, \forall y_3\, \cdots\, \exists/\forall y_n\colon~ ( x,\, y_1,\, y_2,\, y_3,\, \ldots,\, y_n) \in \Problem R \big\},\, \Problem R$ \textnormal{is} \textnormal{a} \textnormal{decidable} \textnormal{re}{\-}\textnormal{la}{\-}\textnormal{tion}$\big\}$.
% and the \emph{class $\Delta_n^0$} is defined as $\Delta_n^0 = \Sigma_n^0 \cap \Pi_n^0$.
Note that we require the values of variables to be drawn from a computable domain.
\emph{Multiple consecutive quantifiers} \emph{of the same type} can be contracted to \emph{one} quantifier of that type, so the number $n$ really refers to the number of necessary \emph{quantifier alternations}.
% rather than to the number of quantifiers used.
A set $\Problem A$ is called \emph{arithmetical}, iff $\Problem A \in \Gamma_n^0$, for $\Gamma \in \{\Sigma,\, \Pi\}$ and $n \in \mathbb N$.
The arithmetical sets form a strict hierarchy, i.e.\ $\Gamma_n^0 \subset \Gamma_{n+1}^0$ 
%and $\Sigma_n^0 \neq \Pi_n^0$ 
holds for $\Gamma \in \{\Sigma,\, \Pi\}$ and $n \geq 0$.
%\begin{align*}
%\begin{array}{r}
%\Sigma_n^0\\\\\\\\
%\Pi_n^0\vphantom{\Big(}
%\end{array} \begin{array}{c}\mathrel{\rotatebox{-35}{\text{\LARGE $\subset$}}}\\\\\mathrel{\rotatebox{+35}{\text{\LARGE $\subset$}}}\end{array}~ \Delta_{n+1}^0 \,\begin{array}{c}\mathrel{\rotatebox{+35}{\text{\LARGE $\subset$}}}\\\vspace{-0.6em}\\\mathrel{\rotatebox{-35}{\text{\LARGE $\subset$}}}\end{array} \begin{array}{l}
%\Sigma_{n+1}^0\\\\\\\\
%\Pi_{n+1}^0\vphantom{\Big(}
%\end{array}
%\end{align*}
%holds for every $n \geq 1$, thus the arithmetical sets form a strict hierarchy.
Furthermore, note that $\Sigma_0^0 = \Pi_0^0$ is exactly the class of the decidable sets and $\Sigma_1^0$ is exactly the class of the computably enumerable sets.
\end{definition}
Next, we recall the concept of many--one reducibility and completeness:
\begin{definition}[Many--One~ Reducibility~ and~ Completeness~ \textnormal{\textbf{\cite{odifreddi1,post44,davis}}}]
Let $\Problem A,\, \Problem B$ be arithmetical sets and let $X$ be some appropriate universe such that $\Problem A,\Problem B \subseteq X$.
$\Problem A$ is called \emph{many--one reducible} (or simply \emph{reducible}) to $\Problem B$, denoted $\Problem A \leqm \Problem B$, iff there exists a computable function $r\colon X \rightarrow X$, such that $\forall\, {x} \in X\colon \big( x \in \Problem A \Longleftrightarrow r( x) \in \Problem B\big)$.
If $r$ is a function such that $r$ reduces $\Problem A$ to $\Problem B$, we denote this by $r\colon \Problem A \leqm \Problem B$.
Note that $\leqm$ is transitive. 

$\Problem A$ is called \emph{$\Gamma_n^0$--complete}, for $\Gamma \in \{\Sigma,\, \Pi\}$, iff both $\Problem A \in \Gamma_n^0$ and $\Problem A$ is \emph{$\Gamma_n^0$--hard}, meaning $\Problem C \leqm \Problem A$, for any set $\Problem C \in \Gamma_n^0$. 
Note that if $\Problem B \in \Gamma_n^0$ and $\Problem A \leqm \Problem B$, then $\Problem A \in \Gamma_n^0$, too.
Furthermore, note that if $\Problem A$ is $\Gamma_n^0$--complete and $\Problem A \leqm \Problem B$, then $\Problem B$ is necessarily $\Gamma_n^0$--hard.
Lastly, note that if $\Problem A$ is $\Sigma_n^0$--complete, then $\Problem{A} \in \Sigma_n^0\setminus \Pi_n^0$. Analogously, if $\Problem A$ is $\Pi_n^0$--complete, then $\Problem{A} \in \Pi_n^0\setminus \Sigma_n^0$.
\end{definition}
In the following, we study the hardness of obtaining covariance approximations both from above and from below.
Furthermore, we are interested in exact values of covariances as well as in deciding whether the covariance is infinite.
In order to formally investigate the arithmetical complexity of these problems, we define four problem sets which relate to upper and lower bounds for covariances and to the question whether the covariance is infinite:
%
%
%\begin{enumerate}
%	\item $\LCOVAR$, which relates to the rational lower bounds of $\Cov{C}{\sigma}{f}{g}$,
%	\item $\RCOVAR$, which relates to the rational upper bounds of $\Cov{C}{\sigma}{f}{g}$,
%	\item $\COVAR$, which relates to the exact value of $\Cov{C}{\sigma}{f}{g}$,
%	\item $\INFCOVAR$, which relates to the question whether $\Cov{C}{\sigma}{f}{g}$ is positively or negatively infinite,
%	\item $\LVAR$, which relates to the rational lower bounds of $\Var{C}{\sigma}{f}$,
%	\item $\RVAR$, which relates to the rational upper bounds of $\Var{C}{\sigma}{f}$,
%	\item $\VAR$, which relates to the exact value of $\Var{C}{\sigma}{f}$, and
%	\item $\INFVAR$, which relates to the question whether $\Var{C}{\sigma}{f}$ is infinite.
%\end{enumerate}
%
%
%Formally, we define those sets as follows:
%
%
\begin{definition}[Approximation Problems for Covariances]
We define the following decision problems:
\begin{align*}
	(C,\, \sigma,\, f,\, g,\, q) \in \LCOVAR \quad&\Longleftrightarrow\quad \Cov{C}{\sigma}{f}{g} > q\\
	(C,\, \sigma,\, f,\, g,\, q) \in \RCOVAR \quad&\Longleftrightarrow\quad \Cov{C}{\sigma}{f}{g} < q\\
	(C,\, \sigma,\, f,\, g,\, q) \in \COVAR \quad&\Longleftrightarrow\quad \Cov{C}{\sigma}{f}{g} = q\\
	(C,\, \sigma,\, f,\, g) \in \INFCOVAR \quad&\Longleftrightarrow\quad \Cov{C}{\sigma}{f}{g} \in \{-\infty,\, +\infty\}%\\[6pt]
%	(C,\, \sigma,\, f,\, q) \in \LVAR \quad&\Longleftrightarrow\quad \Var{C}{\sigma}{f} > q\\
%	(C,\, \sigma,\, f,\, q) \in \RVAR \quad&\Longleftrightarrow\quad \Var{C}{\sigma}{f} < q\\
%	(C,\, \sigma,\, f,\, q) \in \VAR \quad&\Longleftrightarrow\quad \Var{C}{\sigma}{f} = q\\
%	(C,\, \sigma,\, f) \in \INFVAR \quad&\Longleftrightarrow\quad \Var{C}{\sigma}{f} = +\infty~,
\end{align*}
where $C \in \PProgs$, $\sigma \in \States$, $f,g \in \E$, and $q \in \Rats$.\footnote{Note that, for obvious reasons, we restrict to \emph{computable} expectations $f, g$ only.}
\end{definition}
The first fact we establish about the hardness of computing upper and lower bounds of covariances is that this is at most $\Sigma_2^0$--hard, thus not harder than deciding whether a non--probabilistic program, i.e. a program without observations and probabilistic choice, does \emph{not} terminate on all inputs, or deciding whether a probabilistic program terminates after an expected finite number of steps~\cite{odifreddi2,hardness}.
Formally, we establish the following results:
\begin{lemma}
\label{insigma2}
${\LCOVAR}$ and ${\RCOVAR}$
%, ${\LVAR}$, ${\RVAR}$ 
are both in $\Sigma_2^0$.
\end{lemma}
For proving \autoref{insigma2}, we revert to a fact established in \cite{hardness}: All lower bounds for expected outcomes are computably enumerable.
As a consequence, there exists a computable function $\wpk{C}{k}{f}(\sigma)$ that is ascending in $k$, such that for given $C\in\PProgs$, $\sigma \in \States$, and $f\in\E$, we have
\begin{align*}
	\forall\, k \in \Nats\colon \wpk{C}{k}{f}(\sigma) ~&\leq~ \wp{C}{f}(\sigma), \quad\text{and}\\[1ex]
	\quad \sup_{k \in \Nats} \wpk{C}{k}{f}(\sigma) ~&=~ \wp{C}{f}(\sigma)~.
\end{align*}
Intuitively, for every $k \in \Nats$ the function $\wpk{C}{k}{f}(\sigma)$ outputs a lower bound of $\wp{C}{f}(\sigma)$ in ascending order.

Similarly, lower bounds for $\wlp{C}{\boldsymbol{1}}(\sigma)$ can be enumerated.
To see this, note that $\wp{C}{\boldsymbol{1}}(\sigma) = 1$ for any $\observesymbol$--free program $C$ and any state $\sigma$.
$\wp{C}{\boldsymbol{1}}(\sigma)$ can only be decreased by violation of an observation.
Informally, 
\begin{align*}
	\wp{C}{\boldsymbol{1}}(\sigma) ~=~ 1 - \text{``Probability of $C$ 
violating an observation"}~.
\end{align*}
Lower bounds for the latter probability can be enumerated by successively exploring the computation tree of $C$ on input $\sigma$ and accumulating the probability mass of all execution traces that lead to a violation of an observation.
As a consequence, there must exist a computable function $\wlpk{C}{k}{\boldsymbol{1}}(\sigma)$ that is descending in $k$, such that for given $C\in\PProgs$ and $\sigma \in \States$, 
\begin{align*}
	\forall\, k \in \Nats\colon \wlp{C}{\boldsymbol{1}}(\sigma) ~&\leq~ \wlpk{C}{k}{\boldsymbol{1}}(\sigma), \quad\text{and}\\[1ex]
	\wlp{C}{\boldsymbol{1}}(\sigma) ~&=~ \inf_{k \in \Nats} \wlpk{C}{k}{\boldsymbol{1}}(\sigma)~.
\end{align*}
Since $\wpk{C}{k}{f}(\sigma)$ is ascending and $\wlpk{C}{k}{\boldsymbol{1}}(\sigma)$ is descending in $k$, the quotient $\nicefrac{\wpk{C}{k}{f}(\sigma)}{\wlpk{C}{k}{\boldsymbol{1}}(\sigma)}$ is ascending in $k$.
We can now prove \autoref{insigma2}:
\begin{proof}[\autoref{insigma2}]
For $\LCOVAR \in \Sigma_2^0$, consider $(C,\, \sigma,\, f,\, g,\, q) \in \LCOVAR$ iff
\begin{align*}
	 \exists\, k~ \forall\, \ell\colon~ \frac{\wpk{C}{k}{f \cdot g }(\sigma)}{\wlpk{C}{k}{\boldsymbol{1}}(\sigma)} - \frac{\wpk{C}{\ell}{f}(\sigma) \cdot \wpk{C}{\ell}{g}(\sigma)}{\wlpk{C}{\ell}{\boldsymbol{1}}(\sigma)^2} ~>~ q~.
\end{align*}
%
%For proving $\LVAR \in \Sigma_2^0$, recall $\LCOVAR \in \Sigma_2^0$.
%Now $\LVAR$ is reduced to $\LCOVAR$ by setting $r(C,\, \sigma,\, f,\, q)\allowbreak = (C,\, \sigma,\, f,\, f,\, q)$. Hence, $\LVAR \in \Sigma_2^0$.
For the proof for $\RCOVAR$, 
% and $\RVAR$
see \cite{technicalReport} %Appendix \ref{proof:insigma2}.
\qed
\end{proof}
Regarding the hardness of deciding whether a given rational is equal to the covariance and the hardness of deciding non--finiteness of covariances, we establish that this is at most $\Pi_2^0$--hard, thus not harder than deciding whether a non--probabilistic program terminates on all inputs, or deciding whether a probabilistic program does \emph{not} terminate after an expected finite number of steps~\cite{odifreddi2,hardness}.
Formally, we establish the following:
\begin{lemma}
\label{inpi2}
${\COVAR}$ and ${\INFCOVAR}$
%, ${\VAR}$, ${\INFVAR}$
are both in $\Pi_2^0$.
\end{lemma}
So far we provided upper bounds for the computational hardness of solving approximation problems for covariances. 
We now show that these bounds are tight in the sense that these problems are \emph{complete} for their respective level of the arithmetical hierarchy.
For that we need a $\Sigma_2^0$-- and a $\Pi_2^0$--hard problem in order to perform the necessary reductions for proving the hardness results.
Adequate problems are the problem of almost--sure termination and its complement:
\begin{theorem}[Hardness of the Almost--Sure Termination Problem~\cite{hardness}]
Let $C \in \PProgs$ be $\observesymbol$--free. 
Then $C$ \emph{terminates almost--surely} on input $\sigma \in \States$, iff it does so with probability $1$.
%For $\observesymbol$--free programs, 
The problem set $\AST$ is defined as $(C,\, \sigma) \in \AST$ iff $C$ terminates almost--surely on input $\sigma$.
We denote the complement of $\AST$ by $\cAST$.\footnote{Note that by ``complement" we mean not exactly a set theoretic complement but rather all pairs $(C,\, \sigma)$ such that $C$ does not terminate almost--surely on $\sigma$.}
$\AST$ is $\Pi_2^0$--complete and $\cAST$ is $\Sigma_2^0$--complete.
\end{theorem}
By reduction from $\cAST$ we now establish the following hardness results:\!
\begin{lemma}
\label{sigma2hard}
${\LCOVAR}$ and ${\RCOVAR}$
%, ${\LVAR}$, ${\RVAR}$ 
are both $\Sigma_2^0$--hard.
\end{lemma}
\begin{proof}
For proving the $\Sigma_2^0$--hardness of ${\LCOVAR}$, consider the reduction function $r_{\mathcal{L}}(C,\, \sigma) = (C',\, \sigma,\, v,\, v,\, 0)$\footnote{We write $v$ for the expectation that in state $\sigma$ returns $\sigma(v)$.},  with $C' = \COMPOSE{\ASSIGN v 0}{\COMPOSE{\PCHOICE{\SKIP}{\nicefrac 1 2}{C}}{\ASSIGN v 1}}$, where variable $v$ does not occur in $C$.
Now consider the following:
\begin{align*}
	\Cov{C'}{\sigma}{v}{v} 	~=~ &\frac{\wp{C'}{v^2}(\sigma)}{\wlp{C'}{\boldsymbol{1}}(\sigma)} - \frac{\wp{C'}{v}(\sigma)^2}{\wlp{C'}{\boldsymbol{1}}(\sigma)^2}\\[1ex]
						~=~ &\frac{\wp{C'}{v^2}(\sigma)}{1} - \frac{\wp{C'}{v}(\sigma)^2}{1^2}			\tag{$C'$ is $\observesymbol$--free}\\[1ex]
						~=~ &\wp{C'}{v^2}(\sigma) - \wp{C'}{v}(\sigma)^2					\\
	\intertext{Since $v$ does not occur in $C$ and $v$ is set from 0 to 1 if and only if $C'$ has terminated, this is equal to:}
						~=~ &\wp{C'}{\boldsymbol{1}^2}(\sigma) - \wp{C'}{\boldsymbol{1}}(\sigma)^2	\\
						~=~ &\wp{C'}{\boldsymbol{1}}(\sigma) - \wp{C'}{\boldsymbol{1}}(\sigma)^2	
\end{align*}
Note that $\wp{C'}{\boldsymbol{1}}(\sigma)$ is exactly the probability of $C'$ terminating on input $\sigma$.
A plot of this termination probability against the resulting variance is given in \autoref{fig:parabola}.
We observe that $\Cov{C'}{\sigma}{v}{v} = \wp{C'}{\boldsymbol{1}}(\sigma) - \wp{C'}{\boldsymbol{1}}(\sigma)^2 > 0$ iff $C'$ terminates \emph{neither} with probability 0 \emph{nor} with probability 1.
Since, however, $C'$ terminates by construction \emph{at least} with probability $\nicefrac 1 2$, we obtain that $\Cov{C'}{\sigma}{v}{v} > 0$ iff $C'$ terminates with probability less than 1, which is the case iff $C$ terminates with probability less than 1.
Thus $r_{\mathcal{L}}(C,\, \sigma) = (C',\, \sigma,\, v,\, v,\, 0) \in \LCOVAR$ iff $(C,\, \sigma) \in \cAST$. Thus, $r_{\mathcal{L}}\colon \cAST \leqm \LCOVAR$.
Since $\cAST$ is $\Sigma_2^0$--complete, if follows that $\LCOVAR$ is $\Sigma_2^0$--hard.

For the the proof for ${\RCOVAR}$, see~\cite{technicalReport}.% Appendix \ref{proof:sigma2hard}.
\qed
\end{proof}
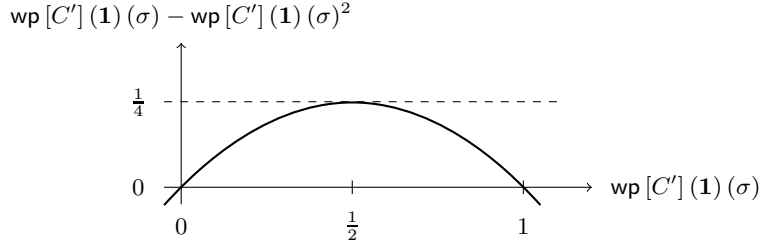
\begin{figure}[t]
	\begin{center}
			\begin{tikzpicture}[scale=4.5]
			\draw[use as bounding box,white] (0,-.16) rectangle (1,0.55);
			% PARABOLA
			\draw[thick] (-0.05,-0.0525) parabola bend (0.5,.25) (1.05,-0.0525);

			% X AXIS
			\draw[->] (-0.05, 0) --  (1.2,0);
			\draw[] (.5, -.02) --  (.5, .02);
			\node (x) at (0.5, -.115) {$\frac 1 2$};
			\draw[] (1, -.02) --  (1, .02);
			\node (x) at (1, -.115) {$1$};
			\node (x) at (0, -.115) {$0$};
			\node (xlabel) at (1.475, 0) {\footnotesize $\wp{C'}{\boldsymbol{1}}(\sigma)$};

			% Y AXIS
			\draw[->] (0, -0.05) --  (0,0.425);
			\draw[dashed] (-.05, .2521) --  (1.1, .2521);
			\node (y) at (-.125, .25) {$\frac 1 4$};
			\node (y) at (-.125, 0) {$0$};
			\node (ylabel) at (0, 0.505) {\footnotesize $\wp{C'}{\boldsymbol{1}}(\sigma) - \wp{C'}{\boldsymbol{1}}(\sigma)^2$};
		\end{tikzpicture}
		\vspace{-1.5\baselineskip}
	\end{center}
	\caption{Plot of the termination probability of a program against the resulting variance.}
	\label{fig:parabola}
\end{figure}
A hardness results for ${\COVAR}$ 
%and ${\VAR}$
is obtained by reduction from $\AST$.
\begin{lemma}
\label{pi2hard1}
${\COVAR}$
% and ${\VAR}$ are both
is $\Pi_2^0$--hard.
\end{lemma}
\begin{proof}
Similar to \autoref{sigma2hard} using $r_{\mathcal{V}}(C,\, \sigma) = \left(C',\, \sigma,\, v,\, v,\, \frac 1 4\right)$, with $C' = \COMPOSE{\ASSIGN v 0}{\COMPOSE{\PCHOICE{\ABORT}{\nicefrac 1 2}{C}}{\ASSIGN v 1}}$.
For details, see~\cite{technicalReport}.% Appendix \ref{proof:pi2hard1}.
\qed
\end{proof}
For a hardness result on ${\INFCOVAR}$
% and ${\INFVAR}$ 
we use the universal halting problem for non--probabilistic programs.
\begin{theorem}[Hardness of the Universal Halting Problem~\textnormal{\textbf{\cite{odifreddi2}}}]
\label{UHPcomplete}
Let $C$ be a non--pro{\-}ba{\-}bi{\-}lis{\-}tic program.
The \emph{universal halting problem} is the problem of deciding whether $C$ terminates on all inputs.
Let $\UHP$ denote the \emph{problem set}, defined as $C \in \UHP$ iff $\forall \sigma \in \States\colon C$ terminates on input $\sigma$.
%We denote by $\boldsymbol\cUHP$ the \emph{complement of $\boldsymbol\UHP$}.
$\UHP$ is $\Pi_2^0$--complete. %and $\cUHP$ is $\Sigma_2^0$--complete.
\end{theorem}
We now establish by reduction from $\UHP$ the remaining hardness result:
\begin{lemma}
\label{pi2hard2}
${\INFCOVAR}$
% and ${\INFVAR}$ are both 
is $\Pi_2^0$--hard.
\end{lemma}
\begin{proof}
For proving the $\Pi_2^0$--hardness of $\INFCOVAR$ we use the reduction function $r_\infty(C) = (C',\, \sigma,\, v,\, v)$, where $\sigma$ is arbitrary but fixed and $C'$ is the  program
\begin{alltt}
\ASSIGN{c}{1}; \ASSIGN{i}{0}; \ASSIGN{x}{0}; \ASSIGN{v}{0}; \ASSIGN{\mathit{term}}{0}; \(\mathit{InitC}\);
while (\(c\) \(\neq\) 0)\{
    \(\mathit{StepC}\); if (\(\mathit{term} = 1\))\{ \(\ASSIGN{v}{2\sp{x}}\); \ASSIGN{i}{i+1}; \ASSIGN{\mathit{term}}{0}; \(\mathit{InitC}\) \};
    \{\ASSIGN{c}{0}\} [0.5] \{\ASSIGN{c}{1}\}; \ASSIGN{x}{x+1} \} \(\textnormal{,}\)
\end{alltt}
where $\mathit{InitC}$ is a non--probabilistic program that initializes a simulation of the program $C$ on input $e(i)$ (where $e\colon \Nats \to \States$ is some computable enumeration of $\States$), and $\mathit{StepC}$ is a non--probabilistic program that does one single (further) step of that simulation and sets \textit{term} to 1 if that step has led to termination of $C$.

Intuitively, the program $C'$ starts by simulating $C$ on input $e(0)$.
During the simulation, it---figuratively speaking---gradually looses interest in further simulating $C$ by tossing a coin after each simulation step to decide whether to continue the simulation or not.
If eventually $C'$ finds that $C$ has terminated on input $e(0)$, it sets the variable $v$ to a number exponential in the number of coin tosses that were made so far, namely to $2^x$.
$C'$ then continues with the same procedure for the next input $e(1)$, and so on.

The variable $x$ keeps track of the number of loop iterations (starting from 1 as the first loop iteration will definitely take place), which equals the number of coin tosses. 
The $x$--th loop iteration takes place with probability $\nicefrac{1}{2^x}$.
The expected value $\Exp{C'}{\sigma}{v}$ is thus given by a series of the form $S = \sum_{i=1}^{\infty} \nicefrac{v_i}{2^i}$, where $v_i = 2^j$ for some $j\in\Nats$.
%%For each iteration of the while--loop, the probability that this iteration will be executed is given by $\nicefrac{1}{2^x}$.
%One loop iteration consists of a constant number of steps $c_1$ in case $Q$ did not halt on input $g_Q(i)$ in the current simulation step.
%Such an iteration therefore contributes $\nicefrac{c_1}{2^x}$ to the expected runtime of the probabilistic program $P$.
%In case $Q$ did halt, a loop iteration takes a constant number of steps $c_2$ plus $2^x$ additional ``cheering" steps.
%Such an iteration therefore contributes $\nicefrac{c_2 + 2^x}{2^x} = \nicefrac{c_2}{2^x} + 1 > 1$ to the expected runtime.
%Overall, the expected runtime of the program $P$ roughly resembles a geometric series with exponentially decreasing summands.
%However, for each time the program $Q$ halts on an input, a summand of the form $\nicefrac{c_2}{2^x} + 1$ appears in this series.
Two cases arise:

\textbf{(1)} $C \in \UHP$, i.e. $C$ terminates on every input.
In that case, $v$ will infinitely often be updated to $2^x$. Therefore, summands of the form $\nicefrac{2^i}{2^i}$ will appear infinitely often in $S$ and so $S$ diverges.
Hence, the expected value of $v$ is infinity and therefore, the variance of $v$ must be infinite as well. Thus, $(C', \sigma,\, v,\, v) \in \INFCOVAR$.

\textbf{(2)} $C \not\in \UHP$, i.e. there exists some input $\sigma'$ with minimal $i \in \Nats$ such that $e(i) = \sigma'$ on which $C$ does not terminate.
In that case, the numerator of all summands of $S$ is upper bounded by some constant $2^j$ and thus $S$ converges.
Boundedness of the $v_i$'s implies that the series $\sum_{i=1}^{\infty} \nicefrac{{v_i}^2}{2^i} = \Exp{C'}{\sigma}{v^2}$ also converges.
Hence, the variance of $v$ is finite and $(C', \sigma,\, v,\, v) \not\in \INFCOVAR$.
%%\textit{\underline{Total Correctness:}} As mentioned in the proof of Lemma \ref{RisSigmacomp}, the program code for $g_Q$ is computable.
%%Also the program code for a universal program capable of simulating any program $Q$ on a given input is computable \cite{kleeneNF}.
%%So in total, the program code for $P$ is computable.
%
%For the full proof including the proof for $\INFCOVAR$, see Appendix \ref{proof:pi2hard2}.
\qed
\end{proof}
Lemmas \ref{insigma2} to \ref{pi2hard2} together directly yield the following completeness results:
\begin{theorem}[The Hardness of Approximating Covariances]
\label{thm:hardness}
	\vspace{-.5\baselineskip}
	\begin{enumerate}
		\item[\textbf{1.}] ${\LCOVAR}$ and ${\RCOVAR}$
		%${\LVAR}$, ${\RVAR}$ are all 
		are both $\Sigma_2^0$--complete.
		\item[\textbf{2.}] ${\COVAR}$ and $\INFCOVAR$
		%, ${\VAR}$ are both 
		are both $\Pi_2^0$--complete.
	\end{enumerate}
\end{theorem}
\begin{remark}[The Hardness of Approximating Variances]
It can be shown that \emph{variance} approximation is not easier than covariance approximation: exactly the same completeness results as in \autoref{thm:hardness} hold for analogous variance approximation problems.
In fact, we have always reduced to approximating a variance for obtaining our hardness results on covariances. \hfill$\triangle$
\end{remark}
As an immediate consequence of \autoref{thm:hardness}, computing both upper and lower bounds for covariances is equally difficult.
This is \emph{contrary to the case for expected values}: 
While computing upper bounds for expected values is also $\Sigma_2^0$--complete, computing lower bounds is $\Sigma_1^0$--complete, thus lower bounds are computably enumerable~\cite{hardness}.
Therefore we can computably enumerate an ascending sequence that converges to the sought--after expected value.
By \autoref{thm:hardness} this is \emph{not possible} for a covariance as $\Sigma_2^0$--sets are in general not computably enumerable.

\autoref{thm:hardness} rules out techniques based on finite loop--unrollings as \emph{complete} approaches for reasoning about the covariances of outcomes of probabilistic programs.
As this is a rather sobering insight, in the next section we will investigate invariant--aided techniques that are complete and can be applied to tackle these approximation problems.

\section{Invariant--Aided Reasoning on Outcome Covariances} \label{sec:invariants}

For straight--line (i.e.\ loop--free) programs, upper and lower bounds for cova{\-}ri{\-}an{\-}ces are obviously computable, e.g.\ by using the decompositions from Definitions \ref{def:cond-ev} and \ref{def:cond-cov}, and the inference rules from \autoref{table:transrules}. 
Problems do arise, however, for loops.
We have seen in the previous section that neither upper nor lower bounds are computably enumerable.
In this section we therefore present an invariant--aided approach for enumerating bounds on covariances of loops.
The underlying principle of such techniques is quite commonly a result due to Park:
\begin{theorem}[Park's Lemma \cite{parkslemma}]
\label{thm:parkslemma}
Let $(D,\, {\sqsubseteq})$ be a complete partial order and $F\colon D \to D$ be continuous.
Then, for all $d \in D$, it holds that $F(d) \sqsubseteq d$ implies $\lfp F 
\sqsubseteq d$, and $d \sqsubseteq F(d)$ implies $d \sqsubseteq \gfp F$.
\end{theorem}
Using this theorem, we can verify in a relatively easy fashion that some element 
is an over--approximation of the least fixed point or an under--approximation 
of the greatest fixed point of a continuous mapping on a complete partial order.
In the following, let $C = \WHILE{B}{C'}$.
In order to exploit Park's Lemma for enumerating bounds on covariances for this while--loop, recall
\begin{align*}
	\Cov{C}{\sigma}{f}{g} 	~=~ &\Exp{C}{\sigma}{f \cdot g} - \Exp{C}{\sigma}{f} \cdot \Exp{C}{\sigma}{g} 	\\[1ex]
							~=~ &\frac{\wp{C}{f \cdot g}(\sigma)}{\wlp{C}{\boldsymbol{1}}(\sigma)} - \frac{\wp{C}{f}(\sigma) \cdot \wp{C}{g}(\sigma)}{\wlp{C}{\boldsymbol{1}}(\sigma)^2}~.
\end{align*}
By inspection of the last line, we can see that for obtaining an over--approxima{\-}tion of $\Cov{C}{\sigma}{f}{g}$, it suffices to over--approximate $\nicefrac{\wp{C'}{f \cdot g}(\sigma)}{\wlp{C'}{\boldsymbol{1}}(\sigma)}$, which can be done by over--approximating $\wp{C'}{f \cdot g}(\sigma)$ and under--approxi{\-}mating $\wlp{C'}{\boldsymbol{1}}(\sigma)$.
Since $\wpsymbol$ ($\wlpsymbol$) of a loop is defined in terms of a least (greatest) fixed point, we can apply Park's Lemma for over--approximating this fraction.
This leads us to the following proof rule:
\begin{theorem}[Invariant--Aided~ Over--Approximation~ of~ Covariances]
\label{thm:outcome-upper}
Let $C = \WHILE{B}{C'}$, $\sigma \in \States$, $f,g\in \E$, $F_h(X) = [\neg 
B]\cdot h + [B] \cdot \wp{C'}{X}$, for any $h \in \E$, and $G(Y) = [\neg B] + [B] \cdot 
\wlp{C'}{Y}$.
Furthermore, let $\widehat{X} \in \E$ and $\widehat{Y} \in \BE$, such that $F_{f \cdot g}\big(\widehat{X}\big) \preceq \widehat{X}$, $\widehat{Y} \preceq G\big(\widehat{Y}\big)$, and $\widehat{Y}(\sigma) > 0$.
Then for all $k \in \Nats$ it holds that\footnote{Here $F_h^k(X)$ stands for $k$--fold application of $F_h$ to $X$.}
\belowdisplayskip=0pt
\begin{align*}
	\Cov{C}{\sigma}{f}{g} ~\leq~ \frac{\widehat{X}(\sigma)}{\widehat{Y}(\sigma)} - \frac{F_{f}^k(\boldsymbol{0})(\sigma) \cdot F_{g}^k(\boldsymbol{0})(\sigma)}{G^k(\boldsymbol{1})(\sigma)^2}~.
\end{align*}
\normalsize
\end{theorem}
By this method we can computably enumerate upper bounds for covariances once appropriate invariants are found.
The catch is that if we choose the invariants, such that $F_{f \cdot g}\big( \widehat{X} \big)(\sigma) < \widehat{X}(\sigma)$ or $\widehat{Y}(\sigma) < G\big( \widehat{Y} \big)(\sigma)$, then the enumeration will \emph{not} get arbitrarily close to the actual covariance.
Note, however, that our method is complete since we could have chosen $\widehat{X} = \lfp F_{f \cdot g}$ and $\widehat{Y} = \gfp G$:
\begin{corollary}[Completeness of \autoref{thm:outcome-upper}] \label{thm:outcome-upper-completeness}
Let $C = \WHILE{B}{C'}$, $\sigma \in \States$, $f,g\in \E$.
Then there exist $\widehat{X} \in \E$ and $\widehat{Y} \in \BE$, such that
\belowdisplayskip=0pt
\begin{align*}
	\inf_{k \in \Nats}~ \frac{\widehat{X}(\sigma)}{\widehat{Y}(\sigma)} - \frac{F_{f}^k(\boldsymbol{0})(\sigma) \cdot F_{g}^k(\boldsymbol{0})(\sigma)}{G^k(\boldsymbol{1})(\sigma)^2} ~=~ \Cov{C}{\sigma}{f}{g}.
\end{align*}
\normalsize
\end{corollary}
By considerations analogous to the ones above, we can formulate dual results for lower bounds.
For details, see~\cite{technicalReport}.
\begin{example}[Application of \autoref{thm:outcome-upper}]
	Reconsider the loop from \autoref{firstexample}.
	For reasoning about the variance of $x$, we pick the invariants 
	\begin{align*}
		\widehat X ~=~ 		&[c \neq 0]\cdot x^2 + [c = 1]\cdot\left( [x \textnormal{ is even}] \cdot \nicefrac{1}{27}\left(9x^2 + 30 x + 41 \right) \right.\\[-1pt]
						&\left.\qquad\qquad\qquad\qquad\qquad\qquad {}+ [x \textnormal{ is odd}] \cdot \nicefrac{2}{27}\left(9x^2 + 12x + 20\right) \right), \quad \text{and}\\[2pt]
		\widehat Y ~=~ 		&[c \neq 0] + [c = 1]\cdot\left( [x \textnormal{ is even}] \cdot \nicefrac 1 3 + [x \textnormal{ is odd}] \cdot \nicefrac 2 3 \right)~,
	\end{align*}
	which satisfy the preconditions of \autoref{thm:outcome-upper}.
	If we enter the loop in a state $\sigma$ with $\sigma(c) = 1$ and $\sigma(x) = 0$, we have $\nicefrac{\widehat{X}(\sigma)}{\widehat{Y}(\sigma)} = \nicefrac{41}{9}$ which is our first upper bound.
	We can now enumerate further upper bounds by doing fixed point iteration on $F_{x}(X) =  [c \neq 1]\cdot x + [c = 1] \cdot \wp{\mathit{loop\,body}}{X} = [c \neq 1]\cdot x + [c = 1] \cdot \frac 1 2 \big([x \textnormal{ is odd}]\cdot X\subst{c}{0} + X\subst{x}{x+1}\big)$ and $G(Y) =  [c \neq 1]+ [c = 1] \cdot \wlp{\mathit{loop\,body}}{Y} = [c \neq 1] + [c = 1] \cdot \frac 1 2 \big([x \textnormal{ is odd}]\cdot Y\subst{c}{0} + Y\subst{x}{x+1}\big)$:
	\begin{align*}
		\frac{41}{9} - \frac{F_{x}^1(\boldsymbol{0})(\sigma)^2}{G^1(\boldsymbol{1})(\sigma)^2} = \frac{41}{9} - \frac{F_{x}^2(\boldsymbol{0})(\sigma)^2}{G^2(\boldsymbol{1})(\sigma)^2} = \frac{41}{9}, \qquad 
		\frac{41}{9} - \frac{F_{x}^3(\boldsymbol{0})(\sigma)^2}{G^3(\boldsymbol{1})(\sigma)^2} = \frac{37}{9}, \qquad \ldots
	\end{align*}
	Finally, this sequence converges to $\nicefrac{41}{9} - \nicefrac{25}{9} = \nicefrac{16}{9}$ as the variance of $x$. \hfill$\triangle$
\end{example}

\section{Reasoning about Run--Time Variances}
In addition to the (co)variance of outcomes we are interested in the variance of the program's \emph{run--time}.
%%%%%
Intuitively, the run--time of a program corresponds to its number of executed operations, where each operation is weighted according to some run--time model.
For simplicity, our run--time model assumes $\SKIP$, guard evaluations and assignments to consume one unit of time. 
Other statements are assumed to consume no time at all.
More elaborated run--time models, e.g. in which the run--time of assignments depends on the size of a given expression, are possible design choices that can easily be integrated in our formalization.
%%%%%

We describe the run--time variance in terms of an operational model  Markov Chain (MC) with rewards.
The model is similar to the ones studied in \cite{mfps,esop16}, but additionally keeps track of the run--time in a dedicated variable $\tau$ which is \emph{not accessible by the program}, but may occur in expectations.
\begin{definition}[Run--Time Expectations]
Let $\Statestau = \{\sigma ~|~ \Vars \dotcup \{\tau\} \to \Rats\}$.
The \emph{set of run--time expectations} is then defined as $\Etau = \left\{t ~\middle|~ t \colon \Statestau \to \PosRealsInf\right\}$.
\end{definition}
A corresponding $\wpsymbol$--style calculus to reason about expected run--times and variances of probabilistic programs is presented afterwards.
We first briefly recall some necessary notions about MCs and refer to~\cite[Ch.\ 10]{katoenbaier} for a comprehensive introduction.
A \emph{Markov} \emph{Chain} is a tuple $\MCSYMBOL = (\MCSTATES, \MCTRANS, s_I, \MCREW)$, where $\MCSTATES$ is a countable set of \emph{states}, $s_I \in \MCSTATES$ is the initial state, $\MCTRANS : \MCSTATES \times \MCSTATES \to [0,1]$ is the \emph{transition} \emph{probability} \emph{function} such that for each state $s \in \MCSTATES$,  $\sum_{s' \in \MCSTATES} \MCTRANS(s,s') \in \{0,1\}$, and $\MCREW : \MCSTATES \to \PosReals$ is a \emph{reward} \emph{function}.
Instead of $\MCTRANS(s,s') = p$, we often write $s \xrightarrow{p} s'$.
A \emph{path} in $\MCSYMBOL$ is a finite or infinite sequence $\pi = s_0 s_1 \ldots$ such that 
$s_i \in S$ and $\MCTRANS(s_i,s_{i+1}) > 0$ for each $i \geq 0$ (where we tacitly assume $\MCTRANS(s_i,s_{i+1}) = 0$ if $\pi$ is a finite path of length $n$ and $i \geq n$).
The \emph{cumulative} \emph{reward} and the probability of a finite path $\hat{\pi} = s_0 \ldots s_n$ are given by
$%\begin{align*}
 \MCREW(\hat{\pi}) ~=~ \sum_{k=0}^{n-1} \MCREW(s_k) 
$ %\qquad \text{and} \qquad
and
$
 \MCPROB{\MCSYMBOL}{\hat{\pi}} ~=~ \prod_{k=0}^{n-1} \MCTRANS(s_k,s_{k+1}).
$ %\end{align*}
These notions are lifted to infinite paths by the standard cylinder set construction (cf.~\cite{katoenbaier}).

Given a set of target states $T \subseteq \MCSTATES$, $\MCREACH{T}$ denotes the set of all paths in $\MCSYMBOL$ reaching a state in $T$ from initial state $s_I$. Analogously, all paths starting in $s_I$ that never reach a state in $T$ are denoted by $\neg \MCREACH{T}$.
The \emph{expected} \emph{reward} that $\MCSYMBOL$ eventually reaches $T$ from a state $s \in \MCSTATES$ is  defined as follows:
\begin{align*}
 \ExpRew{\MCSYMBOL}{\MCREACH{T}} ~=~ \begin{cases}
                                                 \sum_{\pi \in \MCREACH{T}} \MCPROB{\MCSYMBOL}{\pi} \cdot \MCREW(\pi) & ~\text{if}~ \sum_{\pi \in \MCREACH{T}} \MCPROB{\MCSYMBOL}{\pi} = 1 \\
                                                 \infty & ~\text{if}~ \sum_{\pi \in \MCREACH{T}} \MCPROB{\MCSYMBOL}{\pi} < 1.
                                              \end{cases}
\end{align*}
Moreover, the \emph{conditional} \emph{expected} \emph{reward} of $\MCSYMBOL$ reaching $T$ from $s$ under the condition that a set of undesired states $U \subseteq \MCSTATES$ is never reached is given by\footnote{Again, we stick to the convention that $\frac 0 0 = 0$.}
\begin{align*}
 \CondExpRew{\MCSYMBOL}{\MCREACH{T}}{\neg \MCREACH{U}} 
 ~=~
 \frac{
    \ExpRew{\MCSYMBOL}{\MCREACH{T} \cap \neg \MCREACH{U}}
 }{
    \MCPROB{\MCSYMBOL}{\neg \MCREACH{U}}
 }.
\end{align*}
We are now in a position to define an operational model for our probabilistic programming language $\PProgs$. Let $\TERM$ and $\OBSERVEFAIL$ be two special symbols denoting successful termination of a program and failure of an observation, respectively.
\begin{definition}[The Operational MC of a $\PProgs$--Program]\label{def:operational}
 Given a program $C \in \PProgs$, an initial program state $\sigma_0 \in 
\Statestau$ and a post--run--time $t \in \E$, the according \emph{MC} is given 
by $\OPMC{C}{\sigma_0}{t} = (\MCSTATES,\, \MCTRANS,\, s_I,\,\MCREW)$, where
 \begin{itemize}
  \item $\MCSTATES = ((\PProgs \cup \{ \TERM \} \cup \{ \TERM;C ~|~ C \in \PProgs \}) \times \Statestau) ~\cup~ \{ \MCSTATE{\SINK},\, \MCSTATE{\OBSERVEFAIL} \}$,
  \item the transition probability function $\MCTRANS$ is given by the rules in \autoref{fig:transition-function},
  \item $s_I = \MCSTATE{C, \sigma_0}$, and 
  \item $\MCREW : \MCSTATES \to \PosReals$ is the reward function defined by $\MCREW(s) = t(\sigma)$ if $s = \MCSTATE{\TERM, \sigma}$ for some $\sigma \in \Statestau$ and $\MCREW(s) = 0$, otherwise.
%         \begin{align*}
%           \MCREW(s) = \begin{cases}
%                         f(\sigma) & ~\text{if}~ s = \MCSTATE{\TERM, \sigma}, \sigma \in \Statestau \\
%                         0         & ~\text{otherwise}
%                       \end{cases}
%         \end{align*}
 \end{itemize}
\end{definition}
\begin{figure}[tb]
%\vspace{-1.5ex}
\begin{tabular}{p{\textwidth - 1ex}}
           $
             \MCINFRULE
                 {}
                 {\MCTRANSITION{\TERM, \sigma}{\SINK}{1}}
                 {terminated}
           $
           \hfill 
           $
             \MCINFRULE
                 {}
                 {\MCTRANSITION{\SINK}{\SINK}{1}}
                 {sink}
           $
           \\[4ex]
           $
             \MCINFRULE
                 {}
                 {\MCTRANSITION{\EMPTY, \sigma}{\TERM, \sigma}{1}}
                 {empty}
           $
	  \hfill 
           $
             \MCINFRULE
                 {}
                 {\MCTRANSITION{\SKIP, \sigma}{\TERM, \sigma\subst{\tau}{\tau+1}}{1}}
                 {skip}
           $
           \\[4ex]
           $
             \MCINFRULE
                 {}
                 {\MCTRANSITION{\HALT, \sigma}{\SINK}{1}}
                 {halt}
           $
	  \hfill 
           $
             \MCINFRULE
                 {}
                 {\MCTRANSITION{\ASSIGN{x}{E}, \sigma}{\TERM, \sigma[x/E,\, \tau/\tau + 1]}{1}}
                 {assgn}
           $
           \\[4ex]
           $
             \MCINFRULE
                 {\MCTRANSITION{C_1, \sigma}{C_1', \sigma'}{p} \quad 0 < p \leq 1}
                 {\MCTRANSITION{\COMPOSE{C_1}{C_2}, \sigma}{\COMPOSE{C_1'}{C_2}, \sigma'}{p}}
                 {seq-1}
           $
	  \hfill 
           $
             \MCINFRULE
                 {}
                 {\MCTRANSITION{\COMPOSE{\TERM}{C_2}, \sigma}{C_2, \sigma}{1}}
                 {seq-2}
           $
           \\[4ex]
           $
             \MCINFRULE
                 {}
                 {\MCTRANSITION{\PCHOICE{C_1}{p}{C_2}, \sigma}{C_1, \sigma\subst{\tau}{\tau+1}}{p}}
                 {pc-1}
           $
	  \\[4ex]
           $
             \MCINFRULE
                 {}
                 {\MCTRANSITION{\PCHOICE{C_1}{p}{C_2}, \sigma}{C_2, \sigma\subst{\tau}{\tau+1}}{1-p}}
                 {pc-2}
           $
           \\[4ex]
           $
             \MCINFRULE
                 {[B](\sigma) = 1}
                 {\MCTRANSITION{\ITE{B}{C_1}{C_2}, \sigma}{C_1, \sigma\subst{\tau}{\tau+1}}{1}}
                 {if-true}
           $
           \\[4ex]
           $
             \MCINFRULE
                 {[B](\sigma) = 0}
                 {\MCTRANSITION{\ITE{B}{C_1}{C_2}, \sigma}{C_2, \sigma\subst{\tau}{\tau+1}}{1}}
                 {if-false}
           $
           \\[4ex]
           $
             \MCINFRULE
                 {}
                 {\MCTRANSITION{\WHILE{B}{C}, \sigma}{\ITE{B}{\COMPOSE{C}{\WHILE{B}{C}}}{\EMPTY}, \sigma}{1}}
                 {while}
           $
           \\[4ex]
           $
             \MCINFRULE
                 {}
                 {\MCTRANSITION{\ABORT, \sigma}{\ABORT, \sigma}{1}}
                 {diverge}
           $
 	   \\[4ex]
           $
             \MCINFRULE
                 {[B](\sigma) = 1}
                 {\MCTRANSITION{\OBSERVE{B}, \sigma}{\TERM, \sigma\subst{\tau}{\tau+1}}{1}}
                 {observe-true}
           $
           \\[4ex]
           $
             \MCINFRULE
                 {[B](\sigma) = 0}
                 {\MCTRANSITION{\OBSERVE{B}, \sigma}{\OBSERVEFAIL}{1}}
                 {observe-false}
           $
           \hfill
           $
             \MCINFRULE
                 {}
                 {\MCTRANSITION{\OBSERVEFAIL}{\SINK}{1}}
                 {observe-failed}
           $
           \\[4ex]
\end{tabular}
\hrule
\normalsize
  \caption{Rules for defining the transition probability function of the MC of 
a $\PProgs$--program.}
  \label{fig:transition-function}
%\vspace{-1\baselineskip}
\end{figure}
In this construction, $\sigma_0(\tau)$ represents the \emph{post--execution} 
\emph{time} of a program, i.e. the run--time that is added after a program 
finishes its execution. Hence, $\tau$ precisely captures the run--time of a 
program if $\sigma_0(\tau) = 0$.
The rules presented in \autoref{fig:transition-function} defining the transition 
probability function are mostly self--explanatory. 
Since we assume guard evaluations, probabilistic choices, assignments and 
the statement $\SKIP$ to consume one unit of time. Hence, $\tau$ is incremented 
accordingly for each of these statements and remains untouched otherwise. 

\autoref{fig:operational-sketch} sketches the structure of the operational MC 
$\OPMC{C}{\sigma}{t}$. Here, clouds represent a set of states and squiggly 
arrows indicate that a set of states is reachable by one or more paths. Each run 
either terminates successfully (i.e.\ it visits some state 
$\MCSTATE{\TERM,~\sigma'}$), or violates an observation (i.e.\ it visits 
$\MCSTATE{\OBSERVEFAIL}$), or diverges. In the first two cases each run 
eventually ends up in the $\MCSTATE{\SINK}$ state. Note that states of the form 
$\MCSTATE{\TERM,~\sigma'}$ are the only ones that may have a positive reward. 
Furthermore, each of the auxiliary states of the form 
$\MCSTATE{\TERM,~\sigma'}$, $\MCSTATE{\OBSERVEFAIL}$ and $\MCSTATE{\SINK}$ is 
needed to properly deal with $\ABORT$, $\HALT$ and $\OBSERVE{B}$.
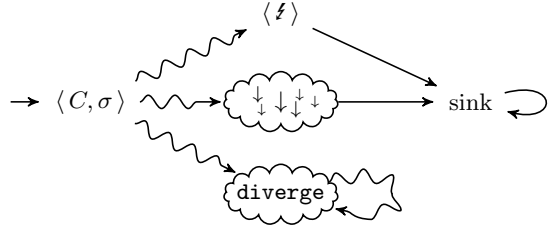
\begin{figure}[t]
  \begin{center}
      \begin{tikzpicture}[->,>=stealth',shorten >=1pt,node distance=2.5cm,semithick,minimum size=1cm]
\tikzstyle{every state}=[draw=none]
  \draw[white, use as bounding box] (-1.2,-1.8) rectangle (6.5,1.5);
   \node [state, initial, initial text=] (init) {$\MCSTATE{C, \sigma}$};  
   \node [cloud, draw=black,cloud puffs=15, cloud puff arc= 150,
        minimum width=1.5cm, minimum height=.75cm, aspect=1] (exit) [right of=init] {$\TERM$};
   \node [state] (bad) [above=0.3cm of exit] {$\MCSTATE{\OBSERVEFAIL}$};
   \node [state] (sink) [right of=exit] {$\SINK$};
   \node [cloud, draw=black,cloud puffs=15, cloud puff arc= 150,
        minimum width=1.5cm, minimum height=.75cm, aspect=1] (diverge) [below=0.5 cm of exit] {$\phantom{\TERM}$};

    \node [] (divergetext) [below=-0.825 cm of diverge] {\small$\mathpzc{\ABORT}$};

   \node [state] (haken1) at (2.2, .1) {\scriptsize $\TERM$};
   \node [state] (haken2) at (2.25, -.1) {\tiny $\TERM$};
    \node [state] (haken3) at (2.7, -.1) {\scriptsize $\TERM$};
   \node [state] (haken4) at (2.75, .1) {\tiny $\TERM$};
   \node [state] (haken5) at (2.95, .0) {\tiny $\TERM$};

  \path [] 
      (init) edge [decorate,decoration={snake, post length=2mm}] (exit)
      (init) edge [decorate,decoration={snake, post length=2mm}] (bad)
      (init) edge [decorate,decoration={snake, post length=2mm}] (diverge)
      (exit) edge [] (sink)
      (bad) edge [] (sink)
      (sink) edge [loop right] (sink)
      (diverge) edge [loop right,decorate,decoration={snake, post length=2mm}] (diverge)
  ;
\end{tikzpicture}
  \end{center}
  \vspace{-1.5\baselineskip}
  \caption{Schematic depiction of the structure of the operational MC $\OPMC{C}{\sigma}{t}$.}
  \label{fig:operational-sketch}
\end{figure}

Since $\tau$ precisely captures the run--time of a program if $\tau$ is initially set to $0$, the \emph{expected} \emph{run--time} of executing $C \in \PProgs$ on input $\sigma \in \Statestau$ with $\sigma(\tau) = 0$ is given by the conditional expected reward of $\OPMC{C}{\sigma}{\tau}$ reaching $\MCSTATE{\SINK}$, given that no observation fails, i.e.~$\Exp{C}{\sigma}{\tau} = \CondExpRew{\OPMC{C}{\sigma}{\tau}}{\MCREACH{\MCSTATE{\SINK}}}{\neg \MCREACH{\MCSTATE{\OBSERVEFAIL}}}$.
Then, in compliance with \autoref{def:covariances}, the \emph{run--time} \emph{variance} $\RtVar{C}{\sigma} $ of $C \in \PProgs$ in state $\sigma \in \Statestau$ with $\sigma(\tau) = 0$ is given by $\Exp{C}{\sigma}{\tau^2} ~-~ \left( \Exp{C}{\sigma}{\tau} \right)^2$ which is
\begin{align*}
	\CondExpRew{\OPMC{C}{\sigma}{\tau^2}}{\MCREACH{\MCSTATE{\SINK}}}{\neg \MCREACH{\MCSTATE{\OBSERVEFAIL}}} - \left( \CondExpRew{\OPMC{C}{\sigma}{\tau}}{\MCREACH{\MCSTATE{\SINK}}}{\neg \MCREACH{\MCSTATE{\OBSERVEFAIL}}} \right)^2~.
\end{align*}
In the following we provide a corresponding $\wpsymbol$--style calculus to reason about expected run--times and run--time variances of probabilistic programs.
A formal definition of the \emph{run--time transformer $\rtsymbol \colon \PProgs \To (\Etau \To \Etau)$} is provided in \autoref{table:transrules} (rightmost column). 
Intuitively, it behaves like $\wpsymbol$ except that a \emph{dedicated run--time variable $\tau$} is updated accordingly for each program statement that consumes time.
%, i.e. for guard evaluations, assignments and $\SKIP$.
In~\cite{esop16}, a transformer for expected run--times without the need for an additional variable $\tau$ is studied. 
However, this approach fails when reasoning about run--time variances since it fails to capture expected squared run--times.
The run--time transformer $\rtsymbol$ precisely captures the notion of expected run--time of our operational model. 
\begin{theorem}[Operational--Denotational Correspondence]
\label{thm:operational:correspondence}
 Let $C \in \PProgs$, $t \in \Etau$, and $\sigma \in \Statestau$. Then 
\belowdisplayskip=0pt
 \begin{align*}
   \CondExpRew{\OPMC{C}{\sigma}{t}}{\MCREACH{\MCSTATE{\SINK}}}{\neg \MCREACH{\MCSTATE{\OBSERVEFAIL}}} ~=~ \frac{\rt{C}{t}(\sigma)}{\wlp{C}{1}(\sigma)}.
 \end{align*}
\normalsize
\end{theorem}

As a result of \autoref{thm:operational:correspondence} we immediately obtain a formal definition of the run--time variance of probabilistic programs in terms of $\rtsymbol$ and $\wlpsymbol$. Formally, the \emph{run--time} \emph{variance} of $C \in \PProgs$ in state $\sigma \in \Statestau$ with $\sigma(\tau) = 0$ is given by
 \begin{align*}
    \RtVar{C}{\sigma} ~=~ & \CondExpRew{\OPMC{C}{\sigma}{\tau^2}}{\MCREACH{\MCSTATE{\SINK}}}{\neg \MCREACH{\MCSTATE{\OBSERVEFAIL}}} \\
                          & ~-~ \left( \CondExpRew{\OPMC{C}{\sigma}{\tau}}{\MCREACH{\MCSTATE{\SINK}}}{\neg \MCREACH{\MCSTATE{\OBSERVEFAIL}}} \right)^2 \\[1ex]
                      ~=~ & \frac{\rt{C}{\tau^2}(\sigma)}{\wlp{C}{\boldsymbol{1}}(\sigma)}
                            ~-~ \frac{\left(\rt{C}{\tau}(\sigma)\right)^2}{\left(\wlp{C}{\boldsymbol{1}}(\sigma)\right)^2}.
 \end{align*}
%
%Similar to outcome (co)variances, neither upper nor lower bounds for run--time variances are computably enumerable. However, 
Since $\rtsymbol$ is continuous (cf.~\cite{technicalReport} for a formal proof), the invariant--aided approach based on Park's Lemma (\autoref{thm:parkslemma}) presented in Section \ref{sec:invariants} is applicable to approximate run--time variances as well.
We present the result for approximating upper bounds only. The dual result for lower bounds is obtained analogously.
\begin{theorem}[Invariant--Aided Over--Approximation of Run--Time Variances] \label{thm:runtime-upper}
  Let $C = \WHILE{B}{C'}$ and $\sigma \in \Statestau$ with $\sigma(\tau) = 0$. 
  Moreover, let $F_h(X) = [\neg B]\cdot h + [B] \cdot \rt{C'}{X}$, and $G(Y) = 
[\neg B] + [B] \cdot \wlp{C'}{Y}$.
  Furthermore, let $\widehat{X} \in \Etau$ and $\widehat{Y} \in \BE$, such that $F_{\tau^2}\big(\widehat{X}\big) \preceq \widehat{X}$, $\widehat{Y} \preceq G\big(\widehat{Y}\big)$, and $\widehat{Y}(\sigma) > 0$.
  Then for each $k \in \Nats$, it holds
\belowdisplayskip=0pt
  \begin{align*}
	  \RtVar{C}{\sigma} ~\leq~ \frac{\widehat{X}(\sigma)}{\widehat{Y}(\sigma)} 
				  ~-~ \left( \frac{F_{\tau}^{k}(\boldsymbol{0})(\sigma)}{G^{k}(\boldsymbol{1})(\sigma)} \right)^2.
  \end{align*} 
\normalsize
\end{theorem}
The proof of \autoref{thm:runtime-upper} is analogous to the proof of \autoref{thm:outcome-upper}.
Again, since it is always possible to choose $\widehat{X} = \lfp F_{\tau^2}$ and $\widehat{Y} = \gfp G$, \autoref{thm:runtime-upper} is complete, i.e. there exist $\widehat{X} \in \Etau$ and $\widehat{Y} \in \BE$ such that
\belowdisplayskip=0pt
\begin{align*}
   \inf_{k \in \Nats}~ \frac{\widehat{X}(\sigma)}{\widehat{Y}(\sigma)} - \left( \frac{F_{\tau}^{k}(\boldsymbol{0})(\sigma)}{G^{k}(\boldsymbol{1})(\sigma)} \right)^2 ~=~ \RtVar{C}{\sigma}.
\end{align*}
\normalsize

\section{Conclusion}
We have studied the computational hardness of obtaining both upper and 
lower bounds on (co)va{\-}ri{\-}an{\-}ce of 
outcomes and established that this is $\Sigma_2^0$--complete. Thus neither upper 
nor lower bounds are computably enumerable.
Furthermore, we have established that deciding whether the (co)variance equals a given rational and deciding whether the covariance is infinite is $\Pi_2^0$--complete.

In the second part of the paper, we continued by presenting a sound and complete invariant--aided approach which allows to computably enumerate upper and lower bounds on (co)variances of while--loops, once appropriate loop--invariants are found. 
Finally, we have shown how this approach can be extended to reason about the 
variance of run--times.

\bibliographystyle{splncs03}
\bibliography{literature}

\clearpage
\appendix
\section{Appendix}
\subsection{Remaining Proof of \autoref{insigma2}}
\label{proof:insigma2}

For proving $\RCOVAR \in \Sigma_2^0$, consider that $(C,\, \sigma,\, f,\, g,\, q) \in \RCOVAR$ iff
\belowdisplayskip=0pt
\begin{align*}
	 \exists\, \delta > 0 ~ \exists\, \ell~ \forall\, k \colon~ \frac{\wpk{C}{k}{f \cdot g}(\sigma)}{\wlpk{C}{k}{\boldsymbol{1}}(\sigma)} - \frac{\wpk{C}{\ell}{f}(\sigma) \cdot \wpk{C}{\ell}{g}(\sigma)}{\wlpk{C}{\ell}{\boldsymbol{1}}(\sigma)^2} ~<~ q - \delta~.
\end{align*}
\normalsize
%
%For $\LVAR \in \Sigma_2^0$, recall $\LCOVAR \in \Sigma_2^0$.
%Now $r$ given by $r(C,\, \sigma,\, f,\, q) = (C,\, \sigma,\, f,\, f,\, q)$, clearly reduces $\LVAR$ to $\LCOVAR$ and so $\LVAR \in \Sigma_2^0$.
%
%Analogously, $r$ as defined above clearly also reduces $\RVAR$ to $\RCOVAR$ and therefore $\RVAR \in \Sigma_2^0$.
\qed

\subsection{Proof of \autoref{inpi2}}
\label{proof:inpi2}

For proving $\COVAR \in \Sigma_2^0$, consider that $(C,\, \sigma,\, f,\, g,\, q) \in \COVAR$ iff
\begin{align*}
	(C,\, \sigma,\, f,\, g,\, q) \not\in \LCOVAR \quad\text{and}\quad (C,\, \sigma,\, f,\, g,\, q) \not\in \RCOVAR~.
\end{align*}
Now by \autoref{insigma2} there must exist \emph{decidable} relations $L$ and $R$, such that $(C,\, \sigma,\, f,\, g,\, q) \in \COVAR$ iff
\begin{align*}
	&\neg \exists\, y_1 \, \forall\, y_2\colon (y_1,\, y_2,\, C, \sigma, f, g, q) \in L \:\wedge\: \neg \exists\, y_1' \, \forall\, y_2'\colon (y_1',\, y_2',\, C, \sigma, f, g, q) \in R\\
	\Longleftrightarrow~&\forall\, y_1 \, \exists\, y_2\colon (y_1,\, y_2,\, C, \sigma, f, g, q) \not\in L \:\wedge\: \forall\, y_1' \, \exists\, y_2'\colon (y_1',\, y_2',\, C, \sigma, f, g, q) \not\in R\\
	\Longleftrightarrow~&\forall\, y_1 \, \forall\, y_1' \, \exists\, y_2 \, \exists\, y_2'\colon (y_1,\, y_2,\, C, \sigma, f, g, q) \not\in L \:\wedge\: (y_1',\, y_2',\, C, \sigma, f, g, q) \not\in R~,
\end{align*}
which is a $\Pi_2$--formula.

For proving $\INFCOVAR \in \Pi_2^0$, consider that $(C,\, \sigma,\, f,\, g,\, q) \in \INFCOVAR$ iff
\begin{align*}
	\neg~\exists\, b_1 \, \exists\, b_2 \, \forall\, k \colon b_1 \leq \frac{\wpk{C}{k}{f \cdot g}(\sigma)}{\wlpk{C}{k}{\boldsymbol{1}}(\sigma)} - \frac{\wpk{C}{k}{f}(\sigma) \cdot \wpk{C}{k}{g}(\sigma)}{\wlpk{C}{k}{\boldsymbol{1}}(\sigma)^2} \leq b_2~.
\end{align*}
The above is the negation of a $\Sigma_2^0$--formula which is equivalent to a $\Pi_2^0$--formula.
%
%For proving $\VAR \in \Pi_2^0$ and $\INFVAR \in \Sigma_2^0$, recall that both $\COVAR \in \Pi_2^0$ and $\INFCOVAR \in \Pi_2^0$.
%Now $r$ as defined in the proof of \autoref{insigma2} clearly also reduces $\VAR$ to $\COVAR$ as well as $\INFVAR$ to $\INFCOVAR$ and therefore $\VAR \in \Pi_2^0$ and $\INFVAR \in \Pi_2^0$.
\qed

\subsection{Remaining Proof of \autoref{sigma2hard}}
\label{proof:sigma2hard}

For proving the $\Sigma_2^0$--hardness of $\RCOVAR$, we reduce the $\Sigma_2^0$--complete $\cAST$ to $\RCOVAR$.
For that, consider the reduction function $r_{\mathcal{R}}(C,\, \sigma) = \left(C',\, \sigma,\, v,\, v,\, \frac 1 4\right)$,  with $C'$ given by
\begin{align*}
	\COMPOSE{\ASSIGN v 0}{\COMPOSE{\PCHOICE{\ABORT}{\nicefrac 1 2}{C}}{\ASSIGN v 1}}~,
\end{align*}
where variable $v$ does not occur in $C$.
We have
\begin{align*}
	\Cov{C'}{\sigma}{v}{v} ~=~  \wp{C'}{\boldsymbol{1}}(\sigma) - \wp{C'}{\boldsymbol{1}}(\sigma)^2~.
\end{align*}
Recall that $\wp{C'}{\boldsymbol{1}}(\sigma)$ is exactly the probability of $C'$ terminating on input $\sigma$.
By reconsidering \autoref{fig:parabola}, we can see that $\Cov{C'}{\sigma}{v}{v} = \wp{C'}{\boldsymbol{1}}(\sigma) - \wp{C'}{\boldsymbol{1}}(\sigma)^2 < \frac{1}{4}$ holds iff $C'$ does not terminate with probability $\nicefrac 1 2$.
Since by construction $C'$ terminates with a probability of at most $\nicefrac 1 2$, it follows that $\Cov{C'}{\sigma}{v}{v} < \frac 1 4$ holds iff $C'$ terminates with probability less than 1, which is the case iff $C$ terminates with probability less than 1.
Thus $r_{\mathcal{R}}(C,\, \sigma) = \left(C',\, \sigma,\, v,\, v,\, \frac 1 4\right) \in \RCOVAR$ iff $(C,\, \sigma) \in \cAST$ and therefore we have $r_{\mathcal{R}}\colon \cAST\allowbreak \leqm \RCOVAR$.
Since $\cAST$ is $\Sigma_2^0$--complete, if follows that $\RCOVAR$ is $\Sigma_2^0$--hard.
%
%For proving that $\RCOVAR$ is $\Sigma_2^0$--hard, recall $\cAST \leqm \RVAR$ and that we have shown $\RVAR \leqm \RCOVAR$ in the proof of \autoref{insigma2}.
%So we obtain $\cAST \leqm \RCOVAR$ by transitivity of $\leqm$.
%Since $\cAST$ is $\Sigma_2^0$--complete, it follows that $\RCOVAR$ is $\Sigma_2^0$--hard.
\qed

\subsection{Proof of \autoref{pi2hard1}}
\label{proof:pi2hard1}

For proving the $\Pi_2^0$--hardness of $\COVAR$, we reduce the $\Pi_2^0$--complete $\AST$ to $\COVAR$.
For that, consider the reduction function $r_{\mathcal{V}}(C,\, \sigma) = \left(C',\, \sigma,\, v,\, v,\, \frac 1 4\right)$,  with $C'$ given by
\begin{align*}
	\COMPOSE{\ASSIGN v 0}{\COMPOSE{\PCHOICE{\ABORT}{\nicefrac 1 2}{C}}{\ASSIGN v 1}}~,
\end{align*}
where variable $v$ does not occur in $C$.
Again, we have
\begin{align*}
	\Cov{C'}{\sigma}{v}{v} ~=~  \wp{C'}{\boldsymbol{1}}(\sigma) - \wp{C'}{\boldsymbol{1}}(\sigma)^2
\end{align*}
(cf.\ proof of \autoref{sigma2hard}).
A plot of the latter is given in \autoref{fig:parabola}.
Recall that $\wp{C'}{\boldsymbol{1}}(\sigma)$ is exactly the probability of $C'$ terminating on input $\sigma$.
We can see that $\Cov{C'}{\sigma}{v}{v} = \wp{C'}{\boldsymbol{1}}(\sigma) - \wp{C'}{\boldsymbol{1}}(\sigma)^2 = \frac{1}{4}$ iff $C'$ terminates with probability $\nicefrac 1 2$.
Since $C'$ terminates at most with probability $\nicefrac 1 2$, we obtain that $\Cov{C'}{\sigma}{v}{v} = \frac 1 4$ iff $C'$ terminates with probability $\nicefrac 1 2$, which is the case iff $C$ terminates almost--surely.
Thus $r_{\mathcal{V}}(C,\, \sigma) = \left(C',\, \sigma,\, v,\, v,\, \frac 1 4\right) \in \COVAR$ iff $(C,\, \sigma) \in \AST$ and therefore $r_{\mathcal{V}}\colon \AST \leqm \COVAR$.
Since $\AST$ is $\Pi_2^0$--complete, we obtain that $\COVAR$ is $\Pi_2^0$--hard.
%
%For proving that $\COVAR$ is $\Pi_2^0$--hard, recall $\AST \leqm \VAR$ and that we have shown $\VAR \leqm \COVAR$ in the proof of \autoref{inpi2}. 
%So we obtain $\AST \leqm \COVAR$ by transitivity of $\leqm$.
%Since $\AST$ is $\Pi_2^0$--complete, we obtain that $\COVAR$ is $\Pi_2^0$--hard.
\qed

\subsection{Proof of \autoref{thm:outcome-upper}}
\label{proof:outcome-upper}

By the precondition of \autoref{thm:outcome-upper} and by \autoref{thm:parkslemma}, $\wp{C}{f \cdot g}(\sigma) \leq \widehat{X}(\sigma)$ and $0 < \widehat{Y}(\sigma) \leq \wlp{C}{\boldsymbol{1}}(\sigma)$.
Therefore
\begin{align*}
	\Exp{C}{\sigma}{f \cdot g} ~=~ \frac{\wp{C}{f \cdot g}(\sigma)}{\wlp{C}{\boldsymbol{1}}(\sigma)} ~\leq~ \frac{\widehat{X}(\sigma)}{\widehat{Y}(\sigma)}~.
\end{align*}
Furthermore, for each $k \in \Nats$ we have
\begin{align*}
	\wpk{C}{k}{f}(\sigma) ~&\leq~ \wp{C}{f}(\sigma)~,\\
	\wpk{C}{k}{g}(\sigma) ~&\leq~ \wp{C}{g}(\sigma)~,\quad\text{ and}\\
	\wpk{C}{k}{\boldsymbol{1}}(\sigma) ~&\geq~ \wlp{C}{\boldsymbol{1}}(\sigma)~,
\end{align*}
which all together yields
\begin{align*}
	&\frac{\wpk{C}{k}{f}(\sigma) \cdot \wpk{C}{k}{g}(\sigma)}{\wlpk{C}{k}{\boldsymbol{1}}(\sigma)^2}\\
	& \qquad~\leq~ \frac{\wp{C}{f}(\sigma) \cdot \wp{C}{g}(\sigma)}{\wlp{C}{\boldsymbol{1}}(\sigma)^2} ~=~ \Exp{C}{\sigma}{f} \cdot \Exp{C}{\sigma}{g}
\end{align*}
An over--approximation of $\Exp{C}{\sigma}{f \cdot g}$ subtracted by an under--approximation of $\Exp{C}{\sigma}{f} \cdot \Exp{C}{\sigma}{g}$ yields then an over--approximation of $\Exp{C}{\sigma}{f \cdot g} -\Exp{C}{\sigma}{f} \cdot \Exp{C}{\sigma}{g} = \Cov{C}{\sigma}{f}{g}$.
\qed

\subsection{Invariant--Aided Under--Approximation of Cova{\-}ri{\-}an{\-}ces}
\label{app:outcome-lower}
\begin{theorem}[Invariant--Aided Under--Approximation of Cova{\-}ri{\-}an{\-}ces]
\label{thm:outcome-lower}
Let $C = \WHILE{B}{C'}$, $\sigma \in \States$, $f,g\in \E$, $F_h(X) = [\neg 
B]\cdot h + [B] \cdot \wp{C'}{X}$, and $G(Y) = [\neg B] + [B] \cdot 
\wlp{C'}{Y}$.
Furthermore, let $\widehat{X}_f, \widehat{X}_g \in \E$ and $\widehat{Y} \in \BE$, such that $F_f(\widehat{X}_f) \preceq \widehat{X}_f$, $F_g(\widehat{X}_g) \preceq \widehat{X}_g$, and $\widehat{Y} \preceq G(\widehat{Y})$.
Then for all $k \in \Nats$ it holds that
\begin{align*}
	\frac{\wpk{C}{k}{f \cdot g}(\sigma)}{\wlpk{C}{k}{\boldsymbol{1}}(\sigma)} - \frac{\widehat{X}_f(\sigma) \cdot \widehat{X}_g(\sigma)}{\widehat{Y}(\sigma)^2} ~\leq~ \Cov{C}{\sigma}{f}{g}~.
\end{align*}
\end{theorem}
\begin{proof}
Analogous to the proof of \autoref{thm:outcome-upper}.
\end{proof}
\begin{corollary}[Completeness of \autoref{thm:outcome-lower}]
Let $C = \WHILE{B}{C'}$, $\sigma \in \States$, $f,g\in \E$.
Then there exist $\widehat{X}_f, \widehat{X}_g \in \E$ and $\widehat{Y} \in \BE$, such that
\begin{align*}
	\sup_{k \in \Nats}~ \frac{\wpk{C}{k}{f \cdot g}(\sigma)}{\wlpk{C}{k}{\boldsymbol{1}}(\sigma)} - \frac{\widehat{X}_f(\sigma) \cdot \widehat{X}_g(\sigma)}{\widehat{Y}(\sigma)^2} ~=~ \Cov{C}{\sigma}{f}{g}~.
\end{align*}
\end{corollary}

\subsection{Proof of \autoref{thm:operational:correspondence}} \label{app:operational:correspondence}

The proof relies on several auxiliary results which are presented first.

\begin{lemma} \label{thm:operational:composition}
 Let $C_1,C_2 \in \PProgs$, $f \in \E$, and $\sigma \in \Statestau$. Then 
 \begin{align*}
   \ExpRew{\OPMC{\COMPOSE{C_1}{C_2}}{\sigma}{f}}{\MCREACH{\MCSTATE{\SINK}}} ~=~ \ExpRew{\OPMC{C_1}{\sigma}{g(C_2,f)}}{\MCREACH{\MCSTATE{\SINK}}},
 \end{align*}
 where
 \begin{align*}
   g(C_2,f) ~=~ \ExpRew{\lambda \rho \,.\, \OPMC{C_2}{\rho}{f}}{\MCREACH{\MCSTATE{\SINK}}}.
 \end{align*}
\end{lemma}

\begin{proof}
 The MC $\OPMC{C}{\sigma}{f}$ is of the following form:
\begin{center}
 \begin{tikzpicture}[->,>=stealth',shorten >=1pt,node distance=2.7cm,semithick,minimum size=1cm]
\tikzstyle{every state}=[draw=none]
  %\draw[white, use as bounding box] (-1.5,-2.5) rectangle (6.8,.2);

   \node [state, initial, initial text=,] (init) {$\MCSTATE{\COMPOSE{C_1}{C_2}, \sigma}$};  
   \node (initLabel) [node distance=0.5cm,gray,below of=init] {$0$};
   
   \node [state,] (termp) [right of = init] {$\MCSTATE{\COMPOSE{\TERM}{C_2}, \sigma'}$};
   \node (termpLabel) [node distance=0.5cm,gray,below of=termp] {$0$};
   
   \node [state,] (termp') [below of = termp] {$\MCSTATE{\COMPOSE{\TERM}{C_2}, \sigma''}$};
   \node (termp'Label) [node distance=0.5cm,gray,below of=termp'] {$0$};
   
   \node [state,] (p) [right of = termp] {$\MCSTATE{C_2, \sigma'}$};
   \node (pLabel) [node distance=0.5cm,gray,below of=p] {$0$};
   
   \node [state,] (q) [right of = termp'] {$\MCSTATE{C_2, \sigma''}$};
   \node (qLabel) [node distance=0.5cm,gray,below of=q] {$0$};
   
   \node [] (dummy) [below= 0.8 cm of init] {$\vdots$};
   
%   \node [state,label={[yshift=0.5cm, gray] 270:$0$}] (sink) [right=.8 of exit] {$\sink$};
    \node [] (dummy1) [on grid, right =2cm of p] {$\ldots$};
    \node [] (dummy2) [on grid, right =2cm of q] {$\ldots$};
  \path [] 
      (init) edge [decorate,decoration={snake, post length=2mm}] (termp)
      (init) edge [decorate,decoration={snake, post length=2mm}] (dummy)
      (init) edge [decorate,decoration={snake, post length=2mm}] (termp')
      (termp) edge [] (p)
      (termp') edge [] (q)
      (p) edge [decorate,decoration={snake, post length=2mm}] (dummy1)
      (q) edge [decorate,decoration={snake, post length=2mm}] (dummy2)
%      (sink) edge [loop right] (sink)
  ;
\end{tikzpicture}
\end{center}
Hence, every path starting in $\MCSTATE{\COMPOSE{C_1}{C_2},\sigma}$ either eventually reaches a state $\MCSTATE{\COMPOSE{\TERM}{C_2}, \sigma'}$ for some $\sigma' \in \Statestau$ and then immediately reaches $\MCSTATE{C_2, \sigma'}$ or diverges, i.e. never reaches $\MCSTATE{\SINK}$. Since $\MCSTATE{C_2, \sigma'}$ is the initial state of MC $\OPMC{C_2}{\sigma'}{f}$, we can transform $\OPMC{\COMPOSE{C_1}{C_2}}{\sigma}{f}$ into an MC $\OPMC{C_1}{\sigma}{g(C_2,f)}$ having the same expected reward by setting 
\begin{align*}
  g(C_2,f) ~=~ \ExpRew{\lambda \rho \,.\, \OPMC{C_2}{\rho}{f}}{\MCREACH{\MCSTATE{\SINK}}}.
\end{align*}
\end{proof}

% \begin{lemma} \label{thm:operational:unrolling}
%  Let $C \in \PProgs$ and $f \in \E$. Then
%  \begin{align*}
%    \rt{\WHILE{B}{C}}{f} ~=~ \rt{\ITE{B}{\COMPOSE{C}{\WHILE{B}{C}}}{\EMPTY}}{f}.
%  \end{align*}
% \end{lemma}
% 
% \begin{proof}
%   Let $F_{f}(X) ~=~ \big([\neg B] \cdot f + [B] \cdot \rt{C}{X}\big)\subst{\tau}{\tau + 1}$. Then
%   \begin{align*}
%          & \rt{\WHILE{B}{C}}{f} \\
%      ~=~ & \lfp X.~ F_{f}(X) \tag{\autoref{table:rtrules}} \\
%      ~=~ & F_f(\lfp X.~ F_{f}(X)) \tag{Def. $\lfp$} \\
%      ~=~ & \big([\neg B] \cdot f + [B] \cdot \rt{C'}{\lfp X.~ F_{f}(X)}\big)\subst{\tau}{\tau + 1} \\
%      ~=~ & \big([\neg B] \cdot f + [B] \cdot \rt{C'}{\rt{\WHILE{B}{C}}{f}}\big)\subst{\tau}{\tau + 1} \tag{\autoref{table:rtrules}} \\
%      ~=~ & \big([\neg B] \cdot \rt{\EMPTY}{f} + [B] \cdot \rt{C}{\rt{\WHILE{B}{C}}{f}}\big)\subst{\tau}{\tau + 1} \tag{$\rt{\EMPTY}{f} = f$} \\
%      ~=~ & \big([\neg B] \cdot \rt{\EMPTY}{f} + [B] \cdot \rt{\COMPOSE{C}{\WHILE{B}{C}}}{f}\big)\subst{\tau}{\tau + 1} \tag{\autoref{table:rtrules}} \\
%      ~=~ & \rt{\ITE{B}{\COMPOSE{C}{\WHILE{B}{C}}}{\EMPTY}}{f}. \tag{\autoref{table:rtrules}} \\
%   \end{align*}
% \end{proof}

\begin{definition}[Bounded $\WHILESYM$ Loops]
 Let $C \in \PProgs$, $f \in \E$, and $k \in \Nats$. Then
 \begin{align*}
    \BOUNDEDWHILE{0}{B}{C} ~=~ & \HALT, ~\text{and} \\
    \BOUNDEDWHILE{k+1}{B}{C} ~=~ & \ITE{B}{\COMPOSE{C}{\BOUNDEDWHILE{k}{B}{C}}}{\EMPTY} \\
 \end{align*}
\end{definition}

To improve readability, let $C' = \WHILE{B}{C}$ and $C_k = \BOUNDEDWHILE{k}{B}{C}$ for the remainder of this section.

\begin{lemma} \label{thm:operational:while-approx}
  Let $C \in \PProgs$ and $f \in \E$. Then
  \begin{align*}
    \sup_{k \in \Nats} \rt{\BOUNDEDWHILE{k}{B}{C}}{f} ~=~ \rt{\WHILE{B}{C}}{f}.
  \end{align*}
\end{lemma}

\begin{proof}
    Let $F_{f}(X) ~=~ \big([\neg B] \cdot f + [B] \cdot \rt{C}{X}\big)\subst{\tau}{\tau + 1}$. Assume, for the moment, that for each $k \in \Nats$, we have
    $\rt{C_k}{f} = F^{k}_{f}(\mathsf{0})$. Then, using Kleene's Fixed Point Theorem, we can establish that
    \begin{align*}
       \sup_{k \in \Nats} \rt{C_k}{f} ~=~ \sup_{k \in \Nats} F^{k}_{f}(\mathsf{0}) ~=~ \lfp X.~F_f(X) ~=~ \rt{C'}{f}.
    \end{align*}
    Hence, it suffices to show $\rt{C_k}{f} = F^{k}_{f}(\mathsf{0})$ for each $k \in \Nats$. We proceed by induction on $k$.
    \paragraph{I.B. $k = 0$}
    \begin{align*}
       \rt{C_0}(f) ~=~ \rt{\HALT}{f} ~=~ \mathsf{0} ~=~ F^{0}_{f}(\mathsf{0}).
    \end{align*}
    
    \paragraph{I.H.} Assume that $\rt{C_k}{f} = F^{k}_{f}(\mathsf{0})$ holds for an arbitrary, fixed $k \in \Nats$.
    
    \paragraph{I.S. $k \mapsto k+1$}
    \begin{align*}
           & \rt{C_{k+1}}{f} \\
       ~=~ & \rt{\ITE{B}{\COMPOSE{C}{C_k}}{\EMPTY}}{f} \tag{Def. $C_{k+1}$} \\
       ~=~ & \big([\neg B] \cdot f + [B] \cdot \rt{C}{\rt{C_k}{f}}\big)\subst{\tau}{\tau + 1} \tag{\autoref{table:transrules}} \\
       ~=~ & \big([\neg B] \cdot f + [B] \cdot \rt{C}{F^{k}_{f}(\mathsf{0})}\big)\subst{\tau}{\tau + 1} \tag{I.H.} \\
       ~=~ & F^{k+1}_{f}(\mathsf{0}) \tag{Def. $F_{f}$}
    \end{align*}
\end{proof}

\begin{lemma} \label{thm:operational:conditional}
 Let $C_1,C_2 \in \PProgs$, $f \in \E$ and $\sigma \in \Statestau$. Then
 \begin{align*}
   & \ExpRew{\OPMC{\ITE{B}{C_1}{C_2}}{\sigma}{f}}{\MCREACH{\MCSTATE{\SINK}}} \\
   ~=~ & [B](\sigma) \cdot \ExpRew{\OPMC{C_1}{\sigma\subst{\tau}{\tau+1}}{f}}{\MCREACH{\MCSTATE{\SINK}}} \\
       & + [\neg B](\sigma) \cdot \ExpRew{\OPMC{C_2}{\sigma\subst{\tau}{\tau+1}}{f}}{\MCREACH{\MCSTATE{\SINK}}}.
 \end{align*}
\end{lemma}

\begin{proof}
  As shown in the figure below, two cases arise. 
  \begin{center}
    \scalebox{0.8}{
\begin{tikzpicture}[->,>=stealth',shorten >=1pt,node distance=2.7cm,semithick,minimum size=1cm]
\tikzstyle{every state}=[draw=none]
  \draw[white, use as bounding box] (-3.5,-4.5) rectangle (10.8,.8);

   \node [state, initial, initial text=,] (init) {$\MCSTATE{\ITE{B}{C_1}{C_2}, \sigma}$};  
   \node (initLabel) [node distance=0.5cm,gray,below of=init] {$0$};

   \node [state,] (p) [node distance = 4.7cm, right of = init] {$\MCSTATE{C_1, \sigma\subst{\tau}{\tau+1}}$};
   \node (pLabel) [node distance=0.5cm,gray,below of=p] {$0$};
   
   \node [state,] (q) [node distance = 3.3cm, below of = init] {$\MCSTATE{C_2, \sigma\subst{\tau}{\tau+1}}$};
   \node (qLabel) [node distance=0.5cm,gray,below of=q] {$0$};
   
    \node [] (dummy1) [on grid, right =3.5cm of p] {$\ldots$};
    \node [] (dummy2) [on grid, right =3.5cm of q] {$\ldots$};
  \path [] 
      (init) edge [] node [below=-0.8cm] {\scriptsize{$[B]$}} (p)
      (initLabel) edge [] node [right] {\scriptsize{$[\neg B]$}} (q)
      (p) edge [decorate,decoration={snake, post length=2mm}] (dummy1)
      (q) edge [decorate,decoration={snake, post length=2mm}] (dummy2)
%      (sink) edge [loop right] (sink)
  ;
\end{tikzpicture}
}
  \end{center}
  If $[B](\sigma) = 1$ then each path $\pi \in \MCREACH{\MCSTATE{\SINK}}$ reaches $\MCSTATE{C_1,\sigma\subst{\tau}{\tau+1}}$ with probability one.
  Conversely, if $[B](\sigma) = 1$ then each path $\pi \in \MCREACH{\MCSTATE{\SINK}}$ reaches $\MCSTATE{C_2,\sigma\subst{\tau}{\tau+1}}$ with probability one.
  These states are the initial states of the MCs $\OPMC{C_1}{\sigma\subst{\tau}{\tau+1}}{f}$ and $\OPMC{C_2}{\sigma\subst{\tau}{\tau+1}}{f}$, respectively.
\end{proof}

\begin{lemma} \label{thm:operational:while-leq}
  Let $C \in \PProgs$, $f \in \E$ and $\sigma \in \Statestau$. Then
  \begin{align*}
    \sup_{k \in \Nats} \ExpRew{\OPMC{\BOUNDEDWHILE{k}{B}{C}}{\sigma}{f}}{\MCREACH{\MCSTATE{\SINK}}} ~\leq~ \ExpRew{\OPMC{\WHILE{B}{C}}{\sigma}{f}}{\MCREACH{\MCSTATE{\SINK}}}.
  \end{align*}
\end{lemma}

\begin{proof}
  We show
  \begin{align*}
    \ExpRew{\OPMC{C_k}{\sigma}{f}}{\MCREACH{\MCSTATE{\SINK}}} ~\leq~ \ExpRew{\OPMC{\WHILE{B}{C}}{\sigma}{f}}{\MCREACH{\MCSTATE{\SINK}}}
  \end{align*}
  for each $k \in \Nats$ by induction on $k$.
  
  \paragraph{I.B. $k=0$}
  \begin{align*}
        & \ExpRew{\OPMC{C_0}{\sigma}{f}}{\MCREACH{\MCSTATE{\SINK}}} \\
    ~=~ & \ExpRew{\OPMC{\HALT}{\sigma}{f}}{\MCREACH{\MCSTATE{\SINK}}} \\
    ~=~ & 0 \leq \ExpRew{\OPMC{\WHILE{B}{C}}{\sigma}{f}}{\MCREACH{\MCSTATE{\SINK}}}.
  \end{align*}
  
  \paragraph{I.H.} Assume for an arbitrary, fixed $k \in \Nats$ that 
  \begin{align*}
    \ExpRew{\OPMC{C_k}{\sigma}{f}}{\MCREACH{\MCSTATE{\SINK}}} ~\leq~ \ExpRew{\OPMC{\WHILE{B}{C}}{\sigma}{f}}{\MCREACH{\MCSTATE{\SINK}}}
  \end{align*}
  
  \paragraph{I.S. $k \mapsto k+1$}
  \begin{align*}
        & \ExpRew{\OPMC{C_{k+1}}{\sigma}{f}}{\MCREACH{\MCSTATE{\SINK}}} \\
    ~=~ & \ExpRew{\OPMC{\ITE{B}{\COMPOSE{C}{C_{k}}}{\EMPTY}}{\sigma}{f}}{\MCREACH{\MCSTATE{\SINK}}}  \tag{Def. $C_{k+1}$} \\
    ~=~ & [B](\sigma) \cdot \ExpRew{\OPMC{\COMPOSE{C}{C_{k}}}{\sigma\subst{\tau}{\tau+1}}{f}}{\MCREACH{\MCSTATE{\SINK}}} \tag{\autoref{thm:operational:conditional}} \\
        & + [\neg B](\sigma) \cdot \ExpRew{\OPMC{\EMPTY}{\sigma\subst{\tau}{\tau+1}}{f}}{\MCREACH{\MCSTATE{\SINK}}} \\
    ~=~ & [B](\sigma) \cdot \ExpRew{\OPMC{C}{\sigma\subst{\tau}{\tau+1}}{g(C_k,f)}}{\MCREACH{\MCSTATE{\SINK}}} \tag{\autoref{thm:operational:composition}} \\
        & + [\neg B](\sigma) \cdot \ExpRew{\OPMC{\EMPTY}{\sigma\subst{\tau}{\tau+1}}{f}}{\MCREACH{\MCSTATE{\SINK}}} \\
    ~\leq~ & [B](\sigma) \cdot \ExpRew{\OPMC{C}{\sigma\subst{\tau}{\tau+1}}{g(C',f)}}{\MCREACH{\MCSTATE{\SINK}}} \tag{I.H.} \\
        & + [\neg B](\sigma) \cdot \ExpRew{\OPMC{\EMPTY}{\sigma\subst{\tau}{\tau+1}}{f}}{\MCREACH{\MCSTATE{\SINK}}} \\        
    ~=~ & [B](\sigma) \cdot \ExpRew{\OPMC{\COMPOSE{C}{C'}}{\sigma\subst{\tau}{\tau+1}}{f}}{\MCREACH{\MCSTATE{\SINK}}} \tag{\autoref{thm:operational:composition}} \\
        & + [\neg B](\sigma) \cdot \ExpRew{\OPMC{\EMPTY}{\sigma\subst{\tau}{\tau+1}}{f}}{\MCREACH{\MCSTATE{\SINK}}} \\    
    ~=~ & \ExpRew{\OPMC{\ITE{B}{\COMPOSE{C}{C'}}{\EMPTY}}{\sigma}{f}}{\MCREACH{\MCSTATE{\SINK}}} \tag{\autoref{thm:operational:conditional}} \\
    ~=~ &  \ExpRew{\OPMC{\WHILE{B}{C}}{\sigma}{f}}{\MCREACH{\MCSTATE{\SINK}}},
  \end{align*}
  where the last step is immediate, because each path of $\OPMC{C'}{\sigma}{f}$ starting in $\MCSTATE{C',\sigma}$ (which has $0$ reward) first reaches state 
  $\MCSTATE{\ITE{B}{\COMPOSE{C}{C'}}{\EMPTY}, \sigma}$ with probability $1$.
\end{proof}

\begin{lemma} \label{thm:operational:while-geq}
  Let $C \in \PProgs$, $f \in \E$ and $\sigma \in \Statestau$. Then
  \begin{align*}
    \sup_{k \in \Nats} \ExpRew{\OPMC{\BOUNDEDWHILE{k}{B}{C}}{\sigma}{f}}{\MCREACH{\MCSTATE{\SINK}}} ~\geq~ \ExpRew{\OPMC{\WHILE{B}{C}}{\sigma}{f}}{\MCREACH{\MCSTATE{\SINK}}}.
  \end{align*}
\end{lemma}

\begin{proof}
   Let $\pi \in \MCREACH{\MCSTATE{\SINK}}$ be a path in $\OPMC{C'}{\sigma}{f}$. Then there exists a finite prefix $\hat{\pi}$ of $\pi$ reaching a state $\MCSTATE{\TERM, \sigma'}$ for some $\sigma' \in \Statestau$ such that $\MCREW(\pi) = \MCREW(\hat{\pi}) > 0$. Since $\hat{\pi}$ is finite, only finitely many states with first component $C'$, say $k$, are visited. We show that a corresponding path $\hat{\pi}'$ with $\MCREW(\hat{\pi}') = \MCREW(\hat{\pi})$ exists in $\OPMC{C_k}{\sigma}{f}$ by induction on $k \geq 1$.
   
   \paragraph{I.B. $k=1$} There exists only one suitable path $\hat{\pi}$ in $\OPMC{C'}{\sigma}{f}$ reaching a state with first component $\TERM$, which is
   \begin{align*}
     \hat{\pi} ~=~ & \MCSTATE{C', \sigma} \MCSTATE{\ITE{B}{\COMPOSE{C}{C'}}{\EMPTY}} \\
                   & \MCSTATE{\EMPTY, \sigma\subst{\tau}{\tau+1}} \MCSTATE{\TERM, \sigma\subst{\tau}{\tau+1}}.
   \end{align*}
   The corresponding path in $\OPMC{C_k}{\sigma}{f}$ with the same reward ($f(\sigma\subst{\tau}{\tau+1})$) is 
   \begin{align*}
     \hat{\pi} ~=~ \MCSTATE{C_1, \sigma} \MCSTATE{\EMPTY, \sigma\subst{\tau}{\tau+1}} \MCSTATE{\TERM, \sigma\subst{\tau}{\tau+1}}.
   \end{align*}
   
   \paragraph{I.H.} For each finite path $\hat{\pi}$ reaching a state with first component $\TERM$ and positive reward in $\OPMC{C'}{\sigma}{f}$ visiting $k > 1$ (for an arbitrary, fixed $k \geq 1$) states with first component $C'$, there exists a path $\hat{\pi}'$ reaching a state with first component $\TERM$ in $\OPMC{C_k}{\sigma}{f}$ such that
   $\MCREW(\hat{\pi}) = \MCREW(\hat{\pi}')$.
   
   \paragraph{I.S. $k \mapsto k+1$} Each finite path $\hat{\pi}$ as described above is of the form
   \begin{align*}
      \hat{\pi} ~=~ & \MCSTATE{C', \sigma} \MCSTATE{\ITE{B}{\COMPOSE{C}{C'}}{\EMPTY}} \\
                    & \MCSTATE{\COMPOSE{C}{C'}, \sigma'} \ldots \MCSTATE{C', \sigma''} \ldots \MCSTATE{\TERM, \sigma'''}
   \end{align*}
   such that $k$ states with first component $C'$ are visited when starting in state $\MCSTATE{C', \sigma'}$. Let $\hat{\pi}_2$ be a path starting in this state. By I.H. there exists a corresponding path $\hat{\pi}_2'$ in $\OPMC{C_k}{\sigma'}{f}$ such that $\MCREW(\hat{\pi}_2) = \MCREW(\hat{\pi}_2')$. Then
   \begin{align*}
      \hat{\pi}' ~=~ \MCSTATE{C_{k+1}, \sigma} \ldots, \MCSTATE{\COMPOSE{\TERM}{C_k}, \sigma'} \hat{\pi}_2'
   \end{align*}
   is a path in $\OPMC{C_{k+1}}{\sigma'}{f}$ with $\MCREW(\hat{\pi}') = \MCREW(\hat{\pi})$.
   Hence, for each finite path $\hat{\pi}$ with positive reward in $\OPMC{C'}{\sigma}{f}$ there exists a corresponding path $\hat{\pi}'$ in $\OPMC{C_k}{\sigma}{f}$ for some $k \in \Nats$. Thus, we include all paths with positive reward in the MC $\OPMC{C'}{\sigma}{f}$ by taking the supremum over $k \in \Nats$ of the expected reward of the MCs $\OPMC{C_k}{\sigma}{f}$.
\end{proof}

Putting \autoref{thm:operational:while-leq} and \autoref{thm:operational:while-geq} together, we immediately obtain

\begin{lemma} \label{thm:operational:while-eq}
  Let $C \in \PProgs$, $f \in \E$ and $\sigma \in \Statestau$. Then
  \begin{align*}
    \sup_{k \in \Nats} \ExpRew{\OPMC{\BOUNDEDWHILE{k}{B}{C}}{\sigma}{f}}{\MCREACH{\MCSTATE{\SINK}}} ~=~ \ExpRew{\OPMC{\WHILE{B}{C}}{\sigma}{f}}{\MCREACH{\MCSTATE{\SINK}}}.
  \end{align*}
\end{lemma}

\begin{lemma} \label{thm:operational:soundness}
 Let $C \in \PProgs$, $f \in \E$, and $\sigma \in \Statestau$. Then 
 \begin{align*}
   \ExpRew{\OPMC{C}{\sigma}{f}}{\MCREACH{\MCSTATE{\SINK}}} ~=~ \rt{C}{f}(\sigma).
 \end{align*}
\end{lemma}

\begin{proof}
  We will make use of a classical substitution lemma stating that $f(\sigma\subst{x}{E}) = f\subst{x}{E}(\sigma)$.
  The proof is by structural induction on the structure of $C$.
  
  \paragraph{The Effectless Program $C = \SKIP$.}
  The MC $\OPMC{\SKIP}{\sigma}{f}$ is of the form
  \begin{center}
   \begin{tikzpicture}[->,>=stealth',shorten >=1pt,node distance=2.5cm,semithick,minimum size=2cm]
\tikzstyle{every state}=[draw=none]
  \draw[white, use as bounding box] (-1.6,-0.85) rectangle (6.5,.55);
  
   \node [state, initial, initial text=,] (init) {$\MCSTATE{\SKIP, \sigma}$};  
   \node (initLabel) [node distance=0.5cm,gray,below of=init] {$0$};
   
   \node [state] (exit) [right of=init] {$\MCSTATE{\TERM, \sigma\subst{\tau}{\tau+1}}$};
   \node (exitLabel) [node distance=0.5cm,gray,below of=exit] {$f(\sigma\subst{\tau}{\tau+1})$};
   
   \node [state] (sink) [right of=exit] {$\MCSTATE{\SINK}$};
   \node (sinkLabel) [node distance=0.5cm,gray,below of=sink] {$0$};

  \path [] 
      (init) edge [] node [above=-0.8cm] {\scriptsize{}} (exit)
      (exit) edge [] node [above=-0.8cm] {\scriptsize{}} (sink)
      (sink) edge [loop right] node [right=-0.8cm] {\scriptsize{}} (sink)
  ;
\end{tikzpicture}
  \end{center}
  Thus, we have $\MCREACH{\MCSTATE{\SINK}} = \{ \hat{\pi} \}$, where
  \begin{align*}
   \hat{\pi} = \MCSTATE{C,\sigma} \MCSTATE{\TERM,\sigma\subst{\tau}{\tau+1}} \MCSTATE{\SINK}.
  \end{align*}
  Moreover, $\MCPROB{\OPMC{C}{\sigma}{f}}{\hat{\pi}} = 1$ and $\MCREW(\hat{\pi}) = f(\sigma\subst{\tau}{\tau+1})$. Thus
  \begin{align*}
       & \ExpRew{\OPMC{C}{\sigma}{f}}{\MCREACH{\MCSTATE{\SINK}}} \\
   ~=~ & \sum_{\pi \in \MCREACH{\MCSTATE{\SINK}}} \MCPROB{\OPMC{C}{\sigma}{f}}{\pi} \cdot \MCREW(\pi) \\
   ~=~ & \MCPROB{\OPMC{C}{\sigma}{f}}{\hat{\pi}} \cdot \MCREW(\hat{\pi}) \\
   ~=~ & 1 \cdot f(\sigma\subst{\tau}{\tau+1}) \\
   ~=~ & f(\sigma\subst{\tau}{\tau+1}) \\
   ~=~ & f\subst{\tau}{\tau+1}(\sigma) \\
   ~=~ & \rt{C}{\sigma}.
  \end{align*}

  \paragraph{The Empty Program $C = \EMPTY$.}
  The MC $\OPMC{\EMPTY}{\sigma}{f}$ is of the form
  \begin{center}
   \begin{tikzpicture}[->,>=stealth',shorten >=1pt,node distance=2.5cm,semithick,minimum size=2cm]
\tikzstyle{every state}=[draw=none]
  \draw[white, use as bounding box] (-1.6,-0.85) rectangle (6.5,.55);
  
   \node [state, initial, initial text=,] (init) {$\MCSTATE{\EMPTY, \sigma}$};  
   \node (initLabel) [node distance=0.5cm,gray,below of=init] {$0$};
   
   \node [state] (exit) [right of=init] {$\MCSTATE{\TERM, \sigma\subst{\tau}{\tau+1}}$};
   \node (exitLabel) [node distance=0.5cm,gray,below of=exit] {$f(\sigma\subst{\tau}{\tau+1})$};
   
   \node [state] (sink) [right of=exit] {$\MCSTATE{\SINK}$};
   \node (sinkLabel) [node distance=0.5cm,gray,below of=sink] {$0$};

  \path [] 
      (init) edge [] node [above=-0.8cm] {\scriptsize{}} (exit)
      (exit) edge [] node [above=-0.8cm] {\scriptsize{}} (sink)
      (sink) edge [loop right] node [right=-0.8cm] {\scriptsize{}} (sink)
  ;
\end{tikzpicture}
  \end{center}
  Thus, we have $\MCREACH{\MCSTATE{\SINK}} = \{ \hat{\pi} \}$, where
  \begin{align*}
   \hat{\pi} = \MCSTATE{C,\sigma} \MCSTATE{\TERM,\sigma} \MCSTATE{\SINK}.
  \end{align*}
  Moreover, $\MCPROB{\OPMC{C}{\sigma}{f}}{\hat{\pi}} = 1$ and $\MCREW(\hat{\pi}) = f(\sigma)$. Thus
  \begin{align*}
       & \ExpRew{\OPMC{C}{\sigma}{f}}{\MCREACH{\MCSTATE{\SINK}}} \\
   ~=~ & \sum_{\pi \in \MCREACH{\MCSTATE{\SINK}}} \MCPROB{\OPMC{C}{\sigma}{f}}{\pi} \cdot \MCREW(\pi) \\
   ~=~ & \MCPROB{\OPMC{C}{\sigma}{f}}{\hat{\pi}} \cdot \MCREW(\hat{\pi}) \\
   ~=~ & 1 \cdot f(\sigma) \\
   ~=~ & f(\sigma) \\
   ~=~ & \rt{C}{\sigma}.
  \end{align*}

  \paragraph{The Diverging Program $C = \ABORT$.}
  The MC $\OPMC{\ABORT}{\sigma}{f}$ is of the form
  \begin{center}
   \begin{tikzpicture}[->,>=stealth',shorten >=1pt,node distance=3.5cm,semithick,minimum size=2cm]
\tikzstyle{every state}=[draw=none]
  \draw[white, use as bounding box] (-1.6,-0.85) rectangle (6.5,.55);
  
   \node [state, initial, initial text=,] (init) {$\MCSTATE{\ABORT, \sigma}$};  
   \node (initLabel) [node distance=0.5cm,gray,below of=init] {$0$};

   \node [state] (sink) [right of=init] {$\MCSTATE{\SINK}$};
   \node (sinkLabel) [node distance=0.5cm,gray,below of=sink] {$0$};

  \path [] 
      (init) edge [loop right] node [above=-0.8cm] {\scriptsize{}} (init)
      (sink) edge [loop right] node [right=-0.8cm] {\scriptsize{}} (sink)
  ;
\end{tikzpicture}
  \end{center}
  Then $\MCREACH{\MCSTATE{\SINK}} = \emptyset$ and thus the probability of reaching $\MCSTATE{\SINK}$ is $0$.
  Thus, we immediately obtain 
  \begin{align*}
       & \ExpRew{\OPMC{C}{\sigma}{f}}{\MCREACH{\MCSTATE{\SINK}}} \\
   ~=~ & \infty \\
   ~=~ & \rt{C}{\sigma}.
  \end{align*}
  
  \paragraph{The Erroneous Program $C = \HALT$.}
  The MC $\OPMC{\HALT}{\sigma}{f}$ is of the form
  
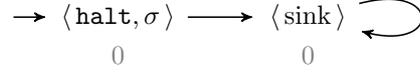
\begin{center}
   \begin{tikzpicture}[->,>=stealth',shorten >=1pt,node distance=2.5cm,semithick,minimum size=2cm]
\tikzstyle{every state}=[draw=none]
  \draw[white, use as bounding box] (-1.6,-0.85) rectangle (6.5,.55);
  
   \node [state, initial, initial text=,] (init) {$\MCSTATE{\HALT, \sigma}$};  
   \node (initLabel) [node distance=0.5cm,gray,below of=init] {$0$};
   
   \node [state] (sink) [right of=init] {$\MCSTATE{\SINK}$};
   \node (sinkLabel) [node distance=0.5cm,gray,below of=sink] {$0$};

  \path [] 
      (init) edge [] node [above=-0.8cm] {\scriptsize{}} (sink)
      (sink) edge [loop right] node [right=-0.8cm] {\scriptsize{}} (sink)
  ;
\end{tikzpicture}
  \end{center}
  Thus, each path $\pi \in \MCREACH{\MCSTATE{\SINK}}$ is of the form
  \begin{align*}
   \hat{\pi} = \MCSTATE{\HALT,\sigma} \MCSTATE{\SINK}
  \end{align*}
  Moreover, $\MCPROB{\OPMC{C}{\sigma}{f}}{\pi} = 1$ and $\MCREW(\pi) = 0$. Then
  \begin{align*}
       & \ExpRew{\OPMC{C}{\sigma}{f}}{\MCREACH{\MCSTATE{\SINK}}} \\
   ~=~ & \sum_{\pi \in \MCREACH{\MCSTATE{\SINK}}} \MCPROB{\OPMC{C}{\sigma}{f}}{\pi} \cdot \MCREW(\pi) \\
   ~=~ & \MCPROB{\OPMC{C}{\sigma}{f}}{\hat{\pi}} \cdot \MCREW(\hat{\pi}) \\
   ~=~ & 1 \cdot 0 \\
   ~=~ & 0 \\
   ~=~ & \rt{C}{\sigma}.
  \end{align*}
  
  \paragraph{The Assignment $C = \ASSIGN{x}{E}$.}
  The MC $\OPMC{\ASSIGN{x}{E}}{\sigma}{f}$ is of the form
  
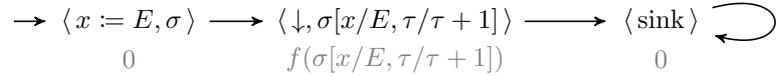
\begin{center}
   \begin{tikzpicture}[->,>=stealth',shorten >=1pt,node distance=3.5cm,semithick,minimum size=2cm]
\tikzstyle{every state}=[draw=none]
  \draw[white, use as bounding box] (-1.6,-0.85) rectangle (6.5,.55);
  
   \node [state, initial, initial text=,] (init) {$\MCSTATE{\ASSIGN{x}{E}, \sigma}$};  
   \node (initLabel) [node distance=0.5cm,gray,below of=init] {$0$};
   
   \node [state] (exit) [right of=init] {$\MCSTATE{\TERM, \sigma[x/E,\tau/\tau+1]}$};
   \node (exitLabel) [node distance=0.5cm,gray,below of=exit] {$f(\sigma[x/E,\tau/\tau+1])$};
   
   \node [state] (sink) [right of=exit] {$\MCSTATE{\SINK}$};
   \node (sinkLabel) [node distance=0.5cm,gray,below of=sink] {$0$};

  \path [] 
      (init) edge [] node [above=-0.8cm] {\scriptsize{}} (exit)
      (exit) edge [] node [above=-0.8cm] {\scriptsize{}} (sink)
      (sink) edge [loop right] node [right=-0.8cm] {\scriptsize{}} (sink)
  ;
\end{tikzpicture}
  \end{center}
  Thus, we have $\MCREACH{\MCSTATE{\SINK}} = \{ \hat{\pi} \}$, where
  \begin{align*}
   \hat{\pi} = \MCSTATE{C,\sigma} \MCSTATE{\TERM,\sigma[x/E,\tau/\tau+1]} \MCSTATE{\SINK}.
  \end{align*}
  Moreover, $\MCPROB{\OPMC{C}{\sigma}{f}}{\hat{\pi}} = 1$ and $\MCREW(\hat{\pi}) = f(\sigma[x/E,\tau/\tau+1])$. Thus
  \begin{align*}
       & \ExpRew{\OPMC{C}{\sigma}{f}}{\MCREACH{\MCSTATE{\SINK}}} \\
   ~=~ & \sum_{\pi \in \MCREACH{\MCSTATE{\SINK}}} \MCPROB{\OPMC{C}{\sigma}{f}}{\pi} \cdot \MCREW(\pi) \\
   ~=~ & \MCPROB{\OPMC{C}{\sigma}{f}}{\hat{\pi}} \cdot \MCREW(\hat{\pi}) \\
   ~=~ & 1 \cdot f(\sigma\subst{\tau}{\tau+1}) \\
   ~=~ & f(\sigma[x/E,\tau/\tau+1]) \\
   ~=~ & f[x/E](\sigma[\tau/\tau+1]) \\
   ~=~ & f[x/E,\tau/\tau+1](\sigma) \\
   ~=~ & \rt{C}{\sigma}.
  \end{align*}

  \paragraph{The Observation $C = \OBSERVE{B}$.}
  The MC $\OPMC{\OBSERVE{B}}{\sigma}{f}$ is of the form
  \begin{center}
   \scalebox{0.8}
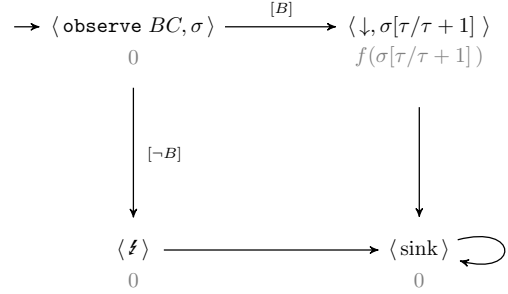
{
\begin{tikzpicture}[->,>=stealth',shorten >=1pt,node distance=2.7cm,semithick,minimum size=1cm]
\tikzstyle{every state}=[draw=none]
  \draw[white, use as bounding box] (-3.5,-4) rectangle (10.8,1);

   \node [state, initial, initial text=,] (init) {$\MCSTATE{\OBSERVE{B}{C}, \sigma}$};  
   \node (initLabel) [node distance=0.5cm,gray,below of=init] {$0$};

   \node [state,] (p) [node distance = 4.7cm, right of = init] {$\MCSTATE{\TERM, \sigma\subst{\tau}{\tau+1}}$};
   \node (pLabel) [node distance=0.5cm,gray,below of=p] {$f(\sigma\subst{\tau}{\tau+1})$};
   
   \node [state,] (q) [node distance = 3.7cm, below of = init] {$\MCSTATE{\OBSERVEFAIL}$};
   \node (qLabel) [node distance=0.5cm,gray,below of=q] {$0$};
   
    \node [] (sink) [on grid, below =3.7cm of p] {$\MCSTATE{\SINK}$};
    \node (qLabel) [node distance=0.5cm,gray,below of=sink] {$0$};
  \path [] 
      (init) edge [] node [below=-0.8cm] {\scriptsize{$[B]$}} (p)
      (initLabel) edge [] node [right] {\scriptsize{$[\neg B]$}} (q)
      (p) edge [] (sink)
      (q) edge [] (sink)
      (sink) edge [loop right] (sink)
  ;
\end{tikzpicture}
}
  \end{center}
  Two cases arise. If $[B](\sigma) = 1$ we have $\MCREACH{\MCSTATE{\SINK}} = \{ \hat{\pi} \}$, where
  \begin{align*}
   \hat{\pi} = \MCSTATE{C,\sigma} \MCSTATE{\TERM,\sigma\subst{\tau}{\tau+1}} \MCSTATE{\SINK}.
  \end{align*}
  Moreover, $\MCPROB{\OPMC{C}{\sigma}{f}}{\hat{\pi}} = 1$ and $\MCREW(\hat{\pi}) = f(\sigma\subst{\tau}{\tau+1})$. Thus
  \begin{align*}
       & \ExpRew{\OPMC{C}{\sigma}{f}}{\MCREACH{\MCSTATE{\SINK}}} \\
   ~=~ & \sum_{\pi \in \MCREACH{\MCSTATE{\SINK}}} \MCPROB{\OPMC{C}{\sigma}{f}}{\pi} \cdot \MCREW(\pi) \\
   ~=~ & \MCPROB{\OPMC{C}{\sigma}{f}}{\hat{\pi}} \cdot \MCREW(\hat{\pi}) \\
   ~=~ & 1 \cdot f(\sigma\subst{\tau}{\tau+1}) \\
   ~=~ & f\subst{\tau}{\tau+1}(\sigma) \\
   ~=~ & [B](\sigma) \cdot f\subst{\tau}{\tau+1}(\sigma) \tag{$[B](\sigma) = 1$} \\
   ~=~ & \rt{C}{\sigma}.
  \end{align*}
  Conversely, if $[B](\sigma) = 0$ we have $\MCREACH{\MCSTATE{\SINK}} = \{ \hat{\pi} \}$, where
  \begin{align*}
   \hat{\pi} = \MCSTATE{C,\sigma} \MCSTATE{\OBSERVEFAIL} \MCSTATE{\SINK}.
  \end{align*}
  Moreover, $\MCPROB{\OPMC{C}{\sigma}{f}}{\hat{\pi}} = 1$ and $\MCREW(\hat{\pi}) = 0$. Thus
  \begin{align*}
       & \ExpRew{\OPMC{C}{\sigma}{f}}{\MCREACH{\MCSTATE{\SINK}}} \\
   ~=~ & \sum_{\pi \in \MCREACH{\MCSTATE{\SINK}}} \MCPROB{\OPMC{C}{\sigma}{f}}{\pi} \cdot \MCREW(\pi) \\
   ~=~ & \MCPROB{\OPMC{C}{\sigma}{f}}{\hat{\pi}} \cdot \MCREW(\hat{\pi}) \\
   ~=~ & 1 \cdot 0 \\
   ~=~ & 0 \\
   ~=~ & [B](\sigma) \cdot f\subst{\tau}{\tau+1}(\sigma) \tag{$[B](\sigma) = 0$} \\
   ~=~ & \rt{C}{\sigma}.
  \end{align*}

\paragraph{The Sequential Composition $C = \COMPOSE{C_1}{C_2}$.}
  \begin{align*}
       & \ExpRew{\OPMC{C}{\sigma}{f}}{\MCREACH{\MCSTATE{\SINK}}} \\
   ~=~ & \ExpRew{\OPMC{C_1}{\sigma}{g(C_2,f)}}{\MCREACH{\MCSTATE{\SINK}}} \tag{\autoref{thm:operational:composition}} \\
   ~=~ & \ExpRew{\OPMC{C_1}{\sigma}{\ExpRew{\lambda \rho \,.\, \OPMC{C_2}{\rho}{f}}{\MCREACH{\MCSTATE{\SINK}}}}}{\MCREACH{\MCSTATE{\SINK}}} \\
   ~=~ & \ExpRew{\OPMC{C_1}{\sigma}{\lambda \rho \,.\, \rt{C_2}{f}(\rho)}}{\MCREACH{\MCSTATE{\SINK}}} \tag{I.H. on $C_2$}\\
   ~=~ & \ExpRew{\OPMC{C_1}{\sigma}{\rt{C_2}{f}}}{\MCREACH{\MCSTATE{\SINK}}} \\
   ~=~ & \rt{C_1}{\rt{C_2}{f}}(\sigma) \tag{I.H. on $C_1$} \\
   ~=~ & \rt{\COMPOSE{C_1}{C_2}}{f}(\sigma).
  \end{align*}

\paragraph{The Conditional $C = \ITE{B}{C_1}{C_2}$.}
  \begin{align*}
       & \ExpRew{\OPMC{C}{\sigma}{f}}{\MCREACH{\MCSTATE{\SINK}}} \\
   ~=~ & [B](\sigma) \cdot \ExpRew{\OPMC{C_1}{\sigma\subst{\tau}{\tau+1}}{f}}{\MCREACH{\MCSTATE{\SINK}}} \tag{\autoref{thm:operational:conditional}}  \\
       & + [\neg B](\sigma) \cdot \ExpRew{\OPMC{C_2}{\sigma\subst{\tau}{\tau+1}}{f}}{\MCREACH{\MCSTATE{\SINK}}} \\
   ~=~ & [B](\sigma) \cdot \rt{C_1}{f}(\sigma\subst{\tau}{\tau+1}) \\
       & + [\neg B](\sigma) \cdot \rt{C_2}{f}(\sigma\subst{\tau}{\tau+1}) \tag{I.H.} \\
   ~=~ & \left([B](\sigma) \cdot \rt{C_1}{f}(\sigma)\right)\subst{\tau}{\tau+1} \tag{$B$ is $\tau$--free} \\
       & + \left([\neg B](\sigma) \cdot \rt{C_2}{f}(\sigma)\right)\subst{\tau}{\tau+1} \\
   ~=~ & \left([B](\sigma) \cdot \rt{C_1}{f}(\sigma) + [\neg B](\sigma) \cdot \rt{C_2}{f}(\sigma)\right)\subst{\tau}{\tau+1} \\
   ~=~ & \rt{C}{\sigma}.
  \end{align*}

\paragraph{The Probabilistic Choice $C = \PCHOICE{C_1}{p}{C_2}$.}
  Each path $\pi \in \MCREACH{\MCSTATE{\SINK}}$ either reaches $\MCSTATE{C_1,\sigma\subst{\tau}{\tau+1}}$ with probability $p$, or
  reaches $\MCSTATE{C_2,\sigma\subst{\tau}{\tau+1}}$ with probability $1-p$.
  These states are the initial states of the MCs $\OPMC{C_1}{\sigma\subst{\tau}{\tau+1}}{f}$ and $\OPMC{C_2}{\sigma\subst{\tau}{\tau+1}}{f}$, respectively. Hence,
  \begin{align*}
       & \ExpRew{\OPMC{C}{\sigma}{f}}{\MCREACH{\MCSTATE{\SINK}}} \\
   ~=~ & p \cdot \ExpRew{\OPMC{C_1}{\sigma\subst{\tau}{\tau+1}}{f}}{\MCREACH{\MCSTATE{\SINK}}} + (1-p) \cdot \ExpRew{\OPMC{C_2}{\sigma\subst{\tau}{\tau+1}}{f}}{\MCREACH{\MCSTATE{\SINK}}} \\
   ~=~ & p \cdot \rt{C_1}{f}(\sigma\subst{\tau}{\tau+1}) + (1-p) \cdot \rt{C_2}{f}(\sigma\subst{\tau}{\tau+1}) \tag{I.H.} \\
   ~=~ & p \cdot \rt{C_1}{f}(\sigma)\subst{\tau}{\tau+1} + (1-p) \cdot \rt{C_2}{f}(\sigma)\subst{\tau}{\tau+1} \\
   ~=~ & \left(p \cdot \rt{C_1}{f}(\sigma) + (1-p) \cdot \rt{C_2}{f}(\sigma)\right)\subst{\tau}{\tau+1} \\
   ~=~ & \rt{C}{\sigma}.
  \end{align*}

\paragraph{The Loop $C = \WHILE{B}{C'}$.}
  For each natural number $k \geq 1$ and $\sigma \in \Statestau$, we have
  \begin{align*}
         & \rt{\BOUNDEDWHILE{k}{B}{C'}}{f}(\sigma) \\
     ~=~ & \rt{\ITE{B}{\COMPOSE{C'}{\BOUNDEDWHILE{k}{B}{C'}}}{\EMPTY}}{f}(\sigma) \\
     ~=~ & ([B](\sigma) \cdot \rt{\COMPOSE{C'}{\BOUNDEDWHILE{k}{B}{C'}}}{f}(\sigma) \\
         & + [\neg B](\sigma) \cdot \rt{\EMPTY}{f}(\sigma))\subst{\tau}{\tau + 1} \\
     ~=~ & ([B](\sigma) \cdot \rt{\COMPOSE{C'}{\BOUNDEDWHILE{k}{B}{C'}}}{f}(\sigma\subst{\tau}{\tau + 1}) \\
         & + [\neg B](\sigma) \cdot \rt{\EMPTY}{f}(\sigma\subst{\tau}{\tau + 1})) \\
     ~=~ & ([B](\sigma) \cdot \ExpRew{\OPMC{\BOUNDEDWHILE{k}{B}{C'}}{\sigma\subst{\tau}{\tau+1}}{f}}{\MCREACH{\MCSTATE{\SINK}}} \tag{I.H.} \\
         & + [\neg B](\sigma) \cdot \ExpRew{\OPMC{\EMPTY}{\sigma\subst{\tau}{\tau+1}}{f}}{\MCREACH{\MCSTATE{\SINK}}}) \\
     ~=~ & \ExpRew{\OPMC{\BOUNDEDWHILE{k}{B}{C'}}{\sigma}{f}}{\MCREACH{\MCSTATE{\SINK}}}.
  \end{align*}

  Together with the already proven equation
  \begin{align*}
     \rt{\HALT}{f}(\sigma) ~=~ \ExpRew{\OPMC{\HALT}{\sigma}{f}}{\MCREACH{\MCSTATE{\SINK}}}
  \end{align*}
  we conclude
  \begin{align*}
       & \rt{\WHILE{B}{C'}}{f}(\sigma) \\
   ~=~ & \sup_{k \in \Nats} \rt{\BOUNDEDWHILE{k}{B}{C'}}{f}(\sigma) \tag{\autoref{thm:operational:while-approx}} \\
   ~=~ & \sup_{k \in \Nats} \ExpRew{\OPMC{\BOUNDEDWHILE{k}{B}{C'}}{\sigma}{f}}{\MCREACH{\MCSTATE{\SINK}}} \\
   ~=~ & \ExpRew{\OPMC{\WHILE{B}{C'}}{\sigma}{f}}{\MCREACH{\MCSTATE{\SINK}}}. \tag{\autoref{thm:operational:while-eq}}
  \end{align*}
\end{proof}

\begin{lemma} \label{thm:operational:sink-no-fail}
 Let $C \in \PProgs$, $f \in \E$, and $\sigma \in \Statestau$. Then 
 \begin{align*}
    \ExpRew{\OPMC{C}{\sigma}{f}}{\MCREACH{\MCSTATE{\SINK}} \cap \neg \MCREACH{\MCSTATE{\OBSERVEFAIL}}} ~=~ \ExpRew{\OPMC{C}{\sigma}{f}}{\MCREACH{\MCSTATE{\SINK}}}.
 \end{align*}
\end{lemma}

\begin{proof}
  Two cases arise. First, assume that the probability of reaching $\MCSTATE{\SINK}$ in $\OPMC{C}{\sigma}{f}$ is $1$. 
  Since no path $\pi \in \MCREACH{\MCSTATE{\SINK}} \cap \MCREACH{\MCSTATE{\OBSERVEFAIL}}$ contains a state of the form $\MCSTATE{\TERM,\sigma}$, we know by \autoref{def:operational}
  that $\MCREW(\pi) = 0$. Hence, 
  \begin{align*}
   \ExpRew{\OPMC{C}{\sigma}{f}}{\MCREACH{\MCSTATE{\SINK}} \cap \MCREACH{\MCSTATE{\OBSERVEFAIL}}} = 0 \tag{$\bigstar$}.
  \end{align*}
  Then
  
  \begin{align*}
        & \ExpRew{\OPMC{C}{\sigma}{f}}{\MCREACH{\MCSTATE{\SINK}}} \\
    ~=~ & \sum_{\pi \in \MCREACH{\MCSTATE{\SINK}}} \MCPROB{\OPMC{C}{\sigma}{f}}{\pi} \cdot \MCREW(\pi) \\
    ~=~ & \sum_{\pi \in \MCREACH{\MCSTATE{\SINK}} \cap \MCREACH{\MCSTATE{\OBSERVEFAIL}}} \MCPROB{\OPMC{C}{\sigma}{f}}{\pi} \cdot \MCREW(\pi) \\
        & + \sum_{\pi \in \MCREACH{\MCSTATE{\SINK}} \cap \neg \MCREACH{\MCSTATE{\OBSERVEFAIL}}} \MCPROB{\OPMC{C}{\sigma}{f}}{\pi} \cdot \MCREW(\pi) \\
    ~=~ & \ExpRew{\OPMC{C}{\sigma}{f}}{\MCREACH{\MCSTATE{\SINK}} \cap \MCREACH{\MCSTATE{\OBSERVEFAIL}}} + \ExpRew{\OPMC{C}{\sigma}{f}}{\MCREACH{\MCSTATE{\SINK}} \cap \neg \MCREACH{\MCSTATE{\OBSERVEFAIL}}} \\
    ~=~ & \ExpRew{\OPMC{C}{\sigma}{f}}{\MCREACH{\MCSTATE{\SINK}} \cap \neg \MCREACH{\MCSTATE{\OBSERVEFAIL}}}. \tag{by $(\bigstar)$}
  \end{align*}  
  
  For the second case assume the probability of reaching $\MCSTATE{\SINK}$ in $\OPMC{C}{\sigma}{f}$ is less than $1$. 
  Then $\MCREACH{\MCSTATE{\SINK} \cap \neg \MCREACH{\MCSTATE{\OBSERVEFAIL}}} ~\subseteq~ \MCREACH{\MCSTATE{\SINK}}$ implies  
  \begin{align*}
      \sum_{\pi \in \MCREACH{\MCSTATE{\SINK}} \cap \neg \MCREACH{\MCSTATE{\OBSERVEFAIL}}} \MCPROB{\OPMC{C}{\sigma}{f}}{\pi}
      ~\leq~
      \sum_{\pi \in \MCREACH{\MCSTATE{\SINK}}} \MCPROB{\OPMC{C}{\sigma}{f}}{\pi} < 1.
  \end{align*} 
  Thus
  \begin{align*}
    \ExpRew{\OPMC{C}{\sigma}{f}}{\MCREACH{\MCSTATE{\SINK}} \cap \neg \MCREACH{\MCSTATE{\OBSERVEFAIL}}} ~=~ \infty ~=~ \ExpRew{\OPMC{C}{\sigma}{f}}{\MCREACH{\MCSTATE{\SINK}}}.
  \end{align*}
\end{proof}

\begin{lemma} \label{thm:operational:wlp-mfps}
  Let $C \in \PProgs$, $f \in \E$, and $\sigma \in \Statestau$. Moreover, let $0 \leq f(\sigma') \leq 1$ for each $\sigma' \in \Statestau$. Then
  \begin{align*}
     \wlp{C}{f}(\sigma) ~=~ \ExpRew{\OPMC{C}{\sigma}{f}}{\MCREACH{\MCSTATE{\SINK}}} + \MCPROB{\OPMC{C}{\sigma}{f}}{\neg \MCREACH{\SINK}}.
  \end{align*} 
\end{lemma}

\begin{proof}
  The proof is by structural induction on the syntactic structure of $\PProgs$ programs and works analogously to the proof of \autoref{thm:operational:soundness}.
  For a detailed proof, we refer to \cite{mfps}.
\end{proof}

\begin{lemma} \label{thm:operational:wlp-obervefail}
  Let $C \in \PProgs$, $f \in \E$, and $\sigma \in \Statestau$. Then
  \begin{align*}
     \wlp{C}{1}(\sigma) ~=~ \MCPROB{\OPMC{C}{\sigma}{f}}{\neg \MCREACH{\OBSERVEFAIL}}.
  \end{align*} 
\end{lemma}

\begin{proof}
  We first observe that each state not visiting $\MCSTATE{\OBSERVEFAIL}$ either
  \begin{enumerate}
   \item visits a state of the form $\MCSTATE{\TERM, \sigma'}$ for some $\sigma' \in \Statestau$ and then immediately reaches $\MCSTATE{\SINK}$, or
   \item visits a state of the form $\MCSTATE{\HALT, \sigma'}$ for some $\sigma' \in \Statestau$ and then immediately reaches $\MCSTATE{\SINK}$, or
   \item diverges and never reaches $\MCSTATE{\SINK}$.
  \end{enumerate}

  Thus,
  \begin{align*}
         & \MCPROB{\OPMC{C}{\sigma}{f}}{\neg \MCREACH{\OBSERVEFAIL}} \\
     ~=~ & \MCPROB{\OPMC{C}{\sigma}{f}}{\MCREACH{\{ \MCSTATE{\TERM, \sigma'} ~|~ \sigma' \in \Statestau \}}}
           + \MCPROB{\OPMC{C}{\sigma}{f}}{\MCREACH{\{ \MCSTATE{\HALT, \sigma'} ~|~ \sigma' \in \Statestau \}}} \\
         & + \MCPROB{\OPMC{C}{\sigma}{f}}{\neg \MCREACH{\SINK}} \\
     ~=~ & \MCPROB{\OPMC{C}{\sigma}{\mathsf{1}}}{\MCREACH{\{ \MCSTATE{\TERM, \sigma'} ~|~ \sigma' \in \Statestau \}}}
           + \MCPROB{\OPMC{C}{\sigma}{\mathsf{1}}}{\MCREACH{\{ \MCSTATE{\HALT, \sigma'} ~|~ \sigma' \in \Statestau \}}} \\
         & + \MCPROB{\OPMC{C}{\sigma}{\mathsf{1}}}{\neg \MCREACH{\SINK}} \\
     ~=~ & \ExpRew{\OPMC{C}{\sigma}{1}}{\MCREACH{\MCSTATE{\SINK}}} + \MCPROB{\OPMC{C}{\sigma}{\mathsf{1}}}{\neg \MCREACH{\SINK}} \tag{\autoref{thm:operational:sink-no-fail} using $(\bigstar)$} \\
     ~=~ & \wlp{C}{1}(\sigma) \tag{\autoref{thm:operational:wlp-mfps}}
  \end{align*}

\end{proof}

\begingroup
\def\thetheorem{\ref{thm:operational:correspondence}}
\begin{theorem}
 Let $C \in \PProgs$, $f \in \E$, and $\sigma \in \Statestau$. Then 
 \begin{align*}
   \CondExpRew{\OPMC{C}{\sigma}{f}}{\MCREACH{\MCSTATE{\SINK}}}{\neg \MCREACH{\MCSTATE{\OBSERVEFAIL}}} ~=~ \frac{\rt{C}{f}(\sigma)}{\wlp{C}{1}(\sigma)}.
 \end{align*}
\end{theorem}
\addtocounter{theorem}{-1}
\endgroup

\begin{proof}
  \begin{align*}
         & \CondExpRew{\OPMC{C}{\sigma}{f}}{\MCREACH{\MCSTATE{\SINK}}}{\neg \MCREACH{\MCSTATE{\OBSERVEFAIL}}} \\[4ex]
     ~=~ & \frac{\ExpRew{\OPMC{C}{\sigma}{f}}{\MCREACH{\MCSTATE{\SINK}} \cap \neg \MCREACH{\MCSTATE{\OBSERVEFAIL}}}}{\MCPROB{\OPMC{C}{\sigma}{f}}{\neg \MCREACH{\MCSTATE{\OBSERVEFAIL}}}}
           \tag{Def. $\textnormal{\textsf{CExpRew}}$} \\[4ex]
    ~=~  & \frac{\ExpRew{\OPMC{C}{\sigma}{f}}{\MCREACH{\MCSTATE{\SINK}}}}{\MCPROB{\OPMC{C}{\sigma}{f}}{\neg \MCREACH{\MCSTATE{\OBSERVEFAIL}}}} \tag{\autoref{thm:operational:sink-no-fail}} \\[4ex]
    ~=~  & \frac{\rt{C}{f}(\sigma)}{\MCPROB{\OPMC{C}{\sigma}{f}}{\neg \MCREACH{\MCSTATE{\OBSERVEFAIL}}}} \tag{\autoref{thm:operational:soundness}} \\[4ex]
    ~=~  & \frac{\rt{C}{f}(\sigma)}{\wlp{C}{1}(\sigma)} \tag{\autoref{thm:operational:wlp-obervefail}}.
  \end{align*}
\end{proof}

\subsection{Continuity of $\rtsymbol$} \label{app:runningtime-continuity}

\begin{lemma}
 For each program $C \in \PProgs$ and every $\omega$--chain $f_1 \preceq f_2 \preceq \ldots$, we have $\rt{C}{\sup_{n} f_n} ~=~ \sup_{n} \rt{C}{f_n}$.
\end{lemma}

\begin{proof}
 Let $f_1 \preceq f_2 \preceq \ldots$ be an $\omega$--chain. We show
 \begin{align*}
    \rt{C}{\sup_{n} f_n} = \sup_{n} \rt{C}{f_n}
 \end{align*}
 by induction on the structure of $C$.
 The makes of of the well known fact that continuous functions are closed under 
substitution\footnote{cf. Theorem 1.3 in D. Scott. "Data types as lattices." 
SIAM Journal on computing 5.3 (1976): 522-587.}.
 
  \paragraph{The Effectless Program $C = \SKIP$}
  \begin{align*}
       & \rt{\SKIP}{\sup_n f_n} \\
   ~=~ & \left( \sup_n f_n \right) \subst{\tau}{\tau+1} \tag{Def. $\rtsymbol$} \\
   ~=~ & \sup_n \left( f_n\subst{\tau}{\tau+1} \right) \\
   ~=~ & \sup_n \rt{\SKIP}{f_n}.
  \end{align*}

  \paragraph{The Empty Program $C = \EMPTY$}
  \begin{align*}
       & \rt{\EMPTY}{\sup_n f_n} \\
   ~=~ & \sup_n f_n \tag{Def. $\rtsymbol$} \\
   ~=~ & \sup_n \rt{\EMPTY}{f_n}. \tag{Def. $\rtsymbol$}
  \end{align*}
  
  \paragraph{The Assignment $C = \ASSIGN{x}{E}$}
  \begin{align*}
       & \rt{\ASSIGN{x}{E}}{\sup_n f_n} \\
   ~=~ & (\sup_n f_n)[x/E, \tau / \tau+1] \tag{Def. $\rtsymbol$} \\
   ~=~ & \left( \sup_n f_n\subst{x}{E} \right)\subst{\tau}{\tau+1} \\
   ~=~ & \sup_n \left( f_n[x/E, \tau / \tau+1] \right) \\
   ~=~ & \sup_n \rt{\ASSIGN{x}{E}}{f_n} \tag{Def. $\rtsymbol$} \\
  \end{align*}

  \paragraph{The Diverging Program $C = \ABORT$}
  \begin{align*}
       & \rt{\ABORT}{\sup_n f_n} \\
   ~=~ & \boldsymbol{\infty} \tag{Def. $\rtsymbol$} \\
   ~=~ & \sup_n \boldsymbol{\infty} \\
   ~=~ & \sup_n \rt{\ABORT}{f_n} \tag{Def. $\rtsymbol$}
  \end{align*}

  \paragraph{The Erroneous Program $C = \HALT$}
  \begin{align*}
       & \rt{\HALT}{\sup_n f_n} \\
   ~=~ & \boldsymbol{0} \tag{Def. $\rtsymbol$} \\
   ~=~ & \sup_n \boldsymbol{0} \\
   ~=~ & \sup_n \rt{\HALT}{f_n} \tag{Def. $\rtsymbol$}
  \end{align*}

  \paragraph{The Observation $C = \OBSERVE{B}$}
  \begin{align*}
       & \rt{\OBSERVE{B}}{\sup_n f_n} \\
   ~=~ & [B] \cdot \left( \sup_n f_n \right)\subst{\tau}{\tau+1} \tag{Def. $\rtsymbol$} \\
   ~=~ & [B] \cdot \left( \sup_n f_n\subst{\tau}{\tau+1} \right) \\
   ~=~ & \sup_n \left( [B] \cdot f_n\subst{\tau}{\tau+1} \right) \\
   ~=~ & \sup_n \rt{\OBSERVE{B}}{f_n}. \tag{Def. $\rtsymbol$}
  \end{align*}

  \paragraph{The Sequential Composition $C = \COMPOSE{C_1}{C_2}$}
  \begin{align*}
       & \rt{\COMPOSE{C_1}{C_2}}{\sup_n f_n} \\
   ~=~ & \rt{C_1}{\rt{C_2}{\sup_n f_n}} \tag{Def. $\rtsymbol$} \\
   ~=~ & \rt{C_1}{\sup_n \rt{C_2}{f_n}} \tag{I.H. on $C_2$} \\
   ~=~ & \sup_n \rt{C_1}{\rt{C_2}{f_n}} \tag{I.H. on $C_1$} \\
   ~=~ & \sup_n \rt{\COMPOSE{C_1}{C_2}}{f_n}.
  \end{align*}

  \paragraph{The Conditional $C = \ITE{B}{C_1}{C_2}$}
   The proof relies on a Monotone Sequence Theorem (MCT) which states that if a sequence $f_1 \preceq f_2 \preceq \ldots$ is a monotonic sequence in $\PosRealsInf$ then $\sup_n f_n = \lim_{n\to\infty} f_n$.
  \begin{align*}
       & \rt{\ITE{B}{C_1}{C_2}}{\sup_n f_n} \\
   ~=~ & \left([B] \cdot \rt{C_1}{\sup_n f_n} + [\neg B] \cdot \rt{C_2}{\sup_n f_n} \right)\subst{\tau}{\tau+1} \tag{Def. $\rtsymbol$} \\
   ~=~ & [B] \cdot \rt{C_1}{\sup_n f_n}\subst{\tau}{\tau+1} + [\neg B] \cdot \rt{C_2}{\sup_n f_n}\subst{\tau}{\tau+1} \\
   ~=~ & [B] \cdot \sup_n \left( \rt{C_1}{f_n}\subst{\tau}{\tau+1} \right) + [\neg B] \cdot \sup_n \left( \rt{C_2}{f_n}\subst{\tau}{\tau+1} \right) \tag{I.H.} \\
   ~=~ & [B] \cdot \lim_{n\to\infty} \left( \rt{C_1}{f_n}\subst{\tau}{\tau+1} \right) + [\neg B] \cdot \lim_{n\to\infty} \left( \rt{C_2}{f_n}\subst{\tau}{\tau+1} \right) \tag{MCT} \\
   ~=~ & \lim_{n\to\infty} \left( [B] \cdot \rt{C_1}{f_n}\subst{\tau}{\tau+1} + [\neg B] \cdot \rt{C_2}{f_n}\subst{\tau}{\tau+1} \right) \\
   ~=~ & \sup_{n} \left( [B] \cdot \rt{C_1}{f_n}\subst{\tau}{\tau+1} + [\neg B] \cdot \rt{C_2}{f_n}\subst{\tau}{\tau+1} \right) \tag{MCT} \\
   ~=~ & \sup_{n} \left( [B] \cdot \rt{C_1}{f_n} + [\neg B] \cdot \rt{C_2}{f_n} \right)\subst{\tau}{\tau+1} \\
   ~=~ & \sup_n \rt{\ITE{B}{C_1}{C_2}}{f_n} \tag{Def. $\rtsymbol$} \\
  \end{align*}

  \paragraph{The Probabilistic Choice $C = \PCHOICE{C_1}{p}{C_2}$}
  The proof relies on a Monotone Sequence Theorem (MCT) which states that if a sequence $f_1 \preceq f_2 \preceq \ldots$ is a monotonic sequence in $\PosRealsInf$ then $\sup_n f_n = \lim_{n\to\infty} f_n$.
  \begin{align*}
       & \rt{\PCHOICE{C_1}{p}{C_2}}{\sup_n f_n} \\
   ~=~ & \left( p \cdot \rt{C_1}{\sup_n f_n} + (1-p) \cdot \rt{C_2}{\sup_n f_n} \right)\subst{\tau}{\tau+1} \tag{Def. $\rtsymbol$} \\
   ~=~ & p \cdot \rt{C_1}{\sup_n f_n}\subst{\tau}{\tau+1} + (1-p) \cdot \rt{C_2}{\sup_n f_n}\subst{\tau}{\tau+1} \\
   ~=~ &  p \cdot \sup_n \left(\rt{C_1}{f_n}\subst{\tau}{\tau+1}\right) + (1-p) \cdot \sup_n \left(\rt{C_2}{f_n}\subst{\tau}{\tau+1}\right) \tag{I.H.} \\
   ~=~ &  p \cdot \lim_{n\to\infty} \left(\rt{C_1}{f_n}\subst{\tau}{\tau+1}\right) + (1-p) \cdot \lim_{n\to\infty} \left(\rt{C_2}{f_n}\subst{\tau}{\tau+1}\right) \tag{MCT} \\
   ~=~ & \lim_{n\to\infty} p \cdot \left(\rt{C_1}{f_n}\subst{\tau}{\tau+1}\right) + (1-p) \cdot \left(\rt{C_2}{f_n}\subst{\tau}{\tau+1}\right) \\
   ~=~ & \sup_n p \cdot \left(\rt{C_1}{f_n}\subst{\tau}{\tau+1}\right) + (1-p) \cdot \left(\rt{C_2}{f_n}\subst{\tau}{\tau+1}\right) \tag{MCT} \\
   ~=~ & \sup_n \left( p \cdot \rt{C_1}{f_n} + (1-p) \cdot \rt{C_2}{f_n} \right)subst{\tau}{\tau+1} \\
   ~=~ & \sup_n \rt{\PCHOICE{C_1}{p}{C_2}}{f_n}. \tag{Def. $\rtsymbol$} \\
  \end{align*}

  \paragraph{The Loop $C = \WHILE{B}{C'}$}
  Let
  \begin{align*}
     F_{f}(X) ~=~ \left( [\neg B] \cdot f + [B] \cdot \rt{C'}{X} \right)\subst{\tau}{\tau+1}.
  \end{align*}
  We make use of two facts concerning continuous functions. Fact 1 states that $F_{\sup_n f_n}(X) = \sup_n F_{f_n}(X)$ and follows from a straightforward reasoning. Fact 2 states that $\sup_n F_{f_n}$ is continuous, because $F_{f_1} \preceq F_{f_2} \preceq \ldots$ is an $\omega$--chain of continuous transformers (note that by I.H. $\rt{C'}{f}$ is continuous). Then
  \begin{align*}
       & \rt{\WHILE{B}{C'}}{\sup_n f_n} \\
   ~=~ & \lfp X ~.~ \left( F_{\sup_n f_n}(X) \right) \tag{Def. $\rtsymbol$} \\
   ~=~ & \sup_{k \in \Nats} F^{k}_{\sup_n f_n}(\boldsymbol{0}) \tag{Tarski's FP-Theorem} \\
   ~=~ & \sup_{k \in \Nats} \left( \sup_n F^{k}_{f_n}\right)(\boldsymbol{0}) \tag{Fact 1} \\
   ~=~ & \sup_{k \in \Nats} \sup_n F_{f_n}^{k}(\boldsymbol{0}) \\
   ~=~ & \sup_n \left( \sup_{k \in \Nats} F_{f_n}^{k}(\boldsymbol{0}) \right) \tag{Fact 2} \\
   ~=~ & \sup_n \rt{\WHILE{B}{C'}}{f_n}. \tag{Def. $\rtsymbol$} \\
  \end{align*}
\end{proof}

\end{document}